\documentclass[journal, onecolumn]{IEEEtran}

\ifCLASSINFOpdf
\else
\fi

\usepackage{graphicx}
\usepackage{graphics}
\usepackage{subfigure}
\usepackage{url}
\usepackage{verbatim} 
\usepackage{multicol}
\usepackage{array}
\usepackage{margins}
\usepackage{psfrag}
\usepackage{xcolor}
\usepackage{cite}
\usepackage{enumerate}

\usepackage{amsthm}
\usepackage{amsmath,amssymb}
\usepackage{amsfonts}

 {\everymath{\scriptscriptstyle\everymath{}}\array}%
 {\endarray}

\newcolumntype{C}[1]{>{\centering\let\newline\\\arraybackslash\hspace{0pt}}m{#1}}
\newcolumntype{L}[1]{>{\raggedright\let\newline\\\arraybackslash\hspace{0pt}}m{#1}}
\newcolumntype{R}[1]{>{\raggedleft\let\newline\\\arraybackslash\hspace{0pt}}m{#1}}

\newcommand{\braceup}[4]{\draw[decorate, decoration={brace, amplitude=5pt},thick] ([xshift=0.5mm,yshift=#3]#1.north west)--([xshift=-0.5mm,yshift=#3]#2.north east) node[midway,anchor=south,outer sep=2mm] {#4}}
\newcommand{\bracedn}[4]{\draw[decorate, decoration={brace, amplitude=5pt},thick] ([xshift=-0.5mm,yshift=#3]#2.south east)--([xshift=0.5mm,yshift=#3]#1.south west) node[midway,anchor=north,outer sep=2mm] {#4}}
\newcommand{\dimup}[4]{\draw[<->,thick] ([xshift=0.5mm,yshift=#3]#1.north west)--([xshift=-0.5mm,yshift=#3]#2.north east) node[midway,anchor=south] {#4}}
\newcommand{\dimdn}[4]{\draw[<->,thick] ([xshift=-0.5mm,yshift=#3]#2.south east)--([xshift=0.5mm,yshift=#3]#1.south west) node[midway,anchor=north] {#4}}
\definecolor{light-gray}{gray}{0.75}

\usepackage{tikz}
\usetikzlibrary{shapes,arrows,shadows,positioning,decorations.pathreplacing,trees,calc}
\tikzstyle{sym} = [draw, thick, rectangle, font=\small, minimum width=4mm, minimum height=4mm, text centered]
\tikzstyle{esym} = [sym, fill=light-gray]
\tikzstyle{usym} = [sym, fill=white]
\tikzstyle{diagbox} = [draw, rectangle, font=\footnotesize, fill=white, text centered, rounded corners]
\tikzstyle{codebox} = [draw, rectangle, font=\footnotesize, minimum height=7mm, fill=white, text centered]
\def\mystrut{\vrule height 0.3cm depth 0.2cm width 0pt}
\def\mystrutt{\vrule height 0.3cm depth 0.2cm width 0pt}

\usetikzlibrary{calc}
\usetikzlibrary{decorations.pathreplacing}
\newcommand{\tikzmark}[1]{\tikz[overlay,remember picture] \node (#1) {};}
\newcommand*{\BraceAmplitude}{0.5em}
\newcommand*{\VerticalOffset}{1.2ex}
\newcommand*{\HorizontalOffset}{0.8em}
\newcommand*{\InsertUnderBrace}[4][]{%
    \begin{tikzpicture}[overlay,remember picture]
\draw [decoration={brace,amplitude=\BraceAmplitude},decorate, thick,draw=blue,text=black,#1]
        ($(#3)+(\HorizontalOffset,-\VerticalOffset)$) -- 
        ($(#2)+(-\HorizontalOffset,-\VerticalOffset)$)
        node [below=\VerticalOffset, midway] {#4};
    \end{tikzpicture}%
}%

\pgfdeclarepatternformonly{north east lines wide}{\pgfqpoint{-1pt}{-1pt}}{\pgfqpoint{4pt}{4pt}}{\pgfqpoint{3pt}{3pt}}%
{
  \pgfsetlinewidth{0.7pt}
  \pgfpathmoveto{\pgfqpoint{0pt}{0pt}}
  \pgfpathlineto{\pgfqpoint{9.1pt}{9.1pt}}
  \pgfusepath{stroke}
}

\tikzset{%
    body/.style={inner sep=0pt,outer sep=0pt,shape=rectangle,draw,thick,pattern=north east lines wide},
    dimen/.style={<->,>=latex,thin,every rectangle node/.style={fill=white,midway,font=\sffamily}},
    symmetry/.style={dashed,thin},
}

\pgfdeclarepatternformonly{soft horizontal lines}{\pgfpointorigin}{\pgfqpoint{100pt}{1pt}}{\pgfqpoint{100pt}{2pt}}%
{
  \pgfsetstrokeopacity{0.3}
  \pgfsetlinewidth{0.5pt}
  \pgfpathmoveto{\pgfqpoint{0pt}{0.5pt}}
  \pgfpathlineto{\pgfqpoint{100pt}{0.5pt}}
  \pgfusepath{stroke}
}

\newcounter{actr}
{\begin{list}{(\alph{actr})}{\usecounter{actr}}}{\end{list}}

\newcounter{ictr}
{\begin{list}{(\roman{ictr})}{\usecounter{ictr}}}{\end{list}}

\newtheorem{remark}{Remark}
\newtheorem{thm}{Theorem}
\newtheorem{lemma}{Lemma}

\newtheorem{defn}{Definition}

\newtheorem{exm}{Example}
\newenvironment{new-proof}[1]
{{\em Proof }:\\}%
{ \noindent\qed }
\newcommand{\defeq}{\stackrel{\Delta}{=}}

\newcommand{\mrm}{\mathrm}

\newcommand{\cC}{{\mathcal{C}}}

\newcommand{\cE}{{\mathcal{E}}}

\newcommand{\bG}{{\mathbf{G}}}

\newcommand{\bH}{{\mathbf{H}}}

\newcommand{\bI}{{\mathbf{I}}}

\newcommand{\bp}{{\mathbf{p}}}

\newcommand{\bq}{{\mathbf{q}}}

\newcommand{\bs}{{\mathbf{s}}}

\newcommand{\cS}{{\mathcal{S}}}

\newcommand{\bu}{{\mathbf{u}}}

\newcommand{\bv}{{\mathbf{v}}}

\newcommand{\bx}{{\mathbf{x}}}

\newcommand{\cX}{{\mathcal{X}}}

\newcommand{\by}{{\mathbf{y}}}

\newcommand{\al}{\alpha}

\newcommand{\eps}{\varepsilon}

\DeclareMathAlphabet{\mathbsf}{OT1}{cmss}{bx}{n}
\DeclareMathAlphabet{\mathssf}{OT1}{cmss}{m}{sl}

\DeclareSymbolFont{bsfletters}{OT1}{cmss}{bx}{n}
\DeclareSymbolFont{ssfletters}{OT1}{cmss}{m}{n}
\DeclareMathSymbol{\bsfGamma}{0}{bsfletters}{'000}
\DeclareMathSymbol{\ssfGamma}{0}{ssfletters}{'000}
\DeclareMathSymbol{\bsfDelta}{0}{bsfletters}{'001}
\DeclareMathSymbol{\ssfDelta}{0}{ssfletters}{'001}
\DeclareMathSymbol{\bsfTheta}{0}{bsfletters}{'002}
\DeclareMathSymbol{\ssfTheta}{0}{ssfletters}{'002}
\DeclareMathSymbol{\bsfLambda}{0}{bsfletters}{'003}
\DeclareMathSymbol{\ssfLambda}{0}{ssfletters}{'003}
\DeclareMathSymbol{\bsfXi}{0}{bsfletters}{'004}
\DeclareMathSymbol{\ssfXi}{0}{ssfletters}{'004}
\DeclareMathSymbol{\bsfPi}{0}{bsfletters}{'005}
\DeclareMathSymbol{\ssfPi}{0}{ssfletters}{'005}
\DeclareMathSymbol{\bsfSigma}{0}{bsfletters}{'006}
\DeclareMathSymbol{\ssfSigma}{0}{ssfletters}{'006}
\DeclareMathSymbol{\bsfUpsilon}{0}{bsfletters}{'007}
\DeclareMathSymbol{\ssfUpsilon}{0}{ssfletters}{'007}
\DeclareMathSymbol{\bsfPhi}{0}{bsfletters}{'010}
\DeclareMathSymbol{\ssfPhi}{0}{ssfletters}{'010}
\DeclareMathSymbol{\bsfPsi}{0}{bsfletters}{'011}
\DeclareMathSymbol{\ssfPsi}{0}{ssfletters}{'011}
\DeclareMathSymbol{\bsfOmega}{0}{bsfletters}{'012}
\DeclareMathSymbol{\ssfOmega}{0}{ssfletters}{'012}

\renewcommand{\defeq}{\triangleq}

\begin{document}
\title{Robust Streaming Erasure Codes based on Deterministic Channel Approximations}


\author{Ahmed Badr,
           Ashish Khisti,        
           Wai-Tian Tan, and
           John Apostolopoulos
\thanks{Ahmed Badr and Ashish Khisti are with the
of Electrical and Computer Engineering Department, University of Toronto, Toronto, ON, Canada.
Wai-Tian Tan and John Apostolopoulos are with the Hewlett-Packard Laboratories, Palo Alto, California. }
}

\maketitle

\begin{abstract}
We study near optimal error correction codes for real-time communication.
In our setup the encoder must operate on an incoming source stream in a sequential manner,
and the decoder must reconstruct each source packet within a fixed playback deadline
of $T$ packets. The underlying channel is a packet erasure channel that can introduce
both burst and isolated losses.

We first consider a class of channels that in any window of length ${T+1}$
introduce either a single erasure burst of a given maximum length $B,$ or a 
certain maximum number $N$ of isolated erasures. We demonstrate that
for a fixed rate and delay, there exists a tradeoff between the achievable values
of $B$ and $N,$ and propose a family of codes that is near optimal with respect
to this tradeoff. We also consider another class of channels that introduce both
a burst  {\em and} an isolated loss in each window of interest and develop the 
associated streaming codes.

All our constructions are based on a layered design and provide
significant improvements over baseline codes
in simulations over the Gilbert-Elliott channel.

\end{abstract}

\IEEEpeerreviewmaketitle

\section{Introduction}
Many emerging multimedia applications require error correction of streaming sources under strict latency constraints.
The transmitter must encode a source stream sequentially and the receiver must decode each source packet within a fixed playback deadline.  Interactive
audio/video conferencing, mobile gaming and cloud-computing are some applications of such systems.
Classical error correction codes are far from ideal in these applications.
The encoders operate on messages in blocks and introduce buffering delays, whereas the 
decoders can only recover missing packets simultaneously without considering the deadline of each packet.
Naturally both the optimal structure and the  fundamental limits of {\em streaming codes} are expected to be different from classical error correction codes. For example it is well known  that the Shannon capacity of an erasure channel only depends on the fraction of erasures. However when delay constraints are imposed, the actual pattern of packet losses also becomes relevant. The decoding delay over channels which introduce burst losses can be very different than  over channels which only introduce isolated losses. In practice  channels such as the Gilbert-Elliott (GE) channel~\cite{gilbert,elliott} introduce both burst and isolated losses. The central question we address in this paper is how to construct streaming codes that significantly outperform classical error correction codes over such channels.

We consider a class of channels that are simplifications of the GE  channel.  Such deterministic models are restricted to  introduce only a certain class  erasure patterns,  which correspond the {\em dominant set} of error events associated with the original channels. We construct near optimal codes for such deterministic approximations and then demonstrate that the resulting codes also yield significant performance gains over the  GE channel in simulations. Note that the GE channel is a two-state Markov model. When the GE channel is in the ``bad state," it introduces a burst-loss whereas when it is in the ``good state" it introduces isolated losses. Therefore a natural approximation to this channel is the following: in any sliding window of a given length $W$, the channel introduces either an erasure burst of a certain maximum length $B$ or up to a certain number $N$ of isolated erasures.   For such channels, we show that for a given rate and delay, there exists an inherent tradeoff between the achievable values of $B$ and $N$. We further propose a class codes --- MiDAS Codes --- that are near optimal with respect to this tradeoff. Our construction  is based on a layered design. We first construct an optimal streaming code for the burst-erasure channel and then introduce another  layer of parity checks for recovery from isolated erasures. 

The above deterministic channel is further improved to capture erasure patterns that appear during the  transition  from the bad state to  good state and vice versa. These channels involve both burst and isolated erasures in the window of interest. We  propose another class of codes --- partial recovery codes (PRC)  --- that recover most of the erased source packets over such channels. We observe in our simulations that when the delay is relatively long, such patterns are dominant and the PRC construction indeed outperforms MiDAS as well as other baseline codes.

In related works,  suitable adaptations of block codes streaming applications have been studied in many prior works see, e.g.,~\cite{roundrobin} and references therein. In reference~\cite{MartinianS04, MartinianT07} a class of optimal streaming codes are proposed for burst erasure channels. Unfortunately these constructions are  sensitive to isolated packet losses. Reference~\cite{MartinianS04} also presents some examples of robust codes using a computer search, but offers limited insights towards a general construction.  In contrast  the present paper proposes a systematic construction of robust streaming codes based on a layered design, establishes fundamental bounds and in the process verifies that some of the robust constructions proposed in~\cite{MartinianS04} are also optimal. Recently connections between streaming codes and  network coding have been studied in~\cite{tekin,leong}. 
However the models in~\cite{tekin,leong} do not aim for robust constructions over the GE channel based on a layered architecture, which is the focus of this paper. 

The rest of the paper is organized as follows. The system model is presented in section~\ref{sec:model} and some baseline codes from earlier works are discussed in section~\ref{prelim}. We introduces our main constructions in sections~\ref{sec:chan-I} and~\ref{sec:chan-2} respectively, and present the simulation results in section~\ref{sec:GE}. 

\section{System Model}
\label{sec:model}
\begin{center}
\begin{figure}[t]
\centering
\begin{minipage}{.48\textwidth}
	\centering
	\resizebox{\columnwidth}{!}{
	\begin{tikzpicture}[node distance=0mm]
		\node[esym]  (x100) {$0$};
		\node[esym, right = of x100]     (x101) {$1$};
		\node[esym, right = of x101]     (x102) {$2$};
		\node[usym, right = of x102]     (x103) {$3$};
		\node[usym, right = of x103]     (x104) {$4$};
		\node[usym, right = of x104]     (x105) {$5$};
		\node[esym, right = of x105]     (x106) {$6$};
		\node[usym, right = of x106]     (x107) {$7$};
		\node[esym, right = of x107]     (x108) {$8$};
		\node[usym, right = of x108]     (x109) {$9$};
		\node[usym, right = of x109]     (x110) {$10$};
		\node[usym, right = of x110]     (x111) {$11$};
		\node[esym, right = of x111]     (x112) {$12$};
		\node[esym, right = of x112]     (x113) {$13$};
		\node[esym, right = of x113]     (x114) {$14$};
		\node[usym, right = of x114]     (x115) {$15$};
		\node      [right = of x115]     (x1end) {$\cdots$};		
	\end{tikzpicture}}
	\caption{An Example of Channel $\rm{I}$:  In any sliding window of length $W=5$ there is either a single erasure burst of length no greater than $B=3$ or up-to $N=2$ erasures.}
	\label{fig:chan-mix}
\end{minipage}%
\hspace{1em}
\begin{minipage}{.48\textwidth}
  \centering
	\resizebox{\columnwidth}{!}{
	\begin{tikzpicture}[node distance=0mm]
		\node[esym]  (x100) {$0$};
		\node[esym, right = of x100]     (x101) {$1$};
		\node[esym, right = of x101]     (x102) {$2$};
		\node[usym, right = of x102]     (x103) {$3$};
		\node[esym, right = of x103]     (x104) {$4$};
		\node[usym, right = of x104]     (x105) {$5$};
		\node[usym, right = of x105]     (x106) {$6$};
		\node[esym, right = of x106]     (x107) {$7$};
		\node[esym, right = of x107]     (x108) {$8$};
		\node[esym, right = of x108]     (x109) {$9$};
		\node[esym, right = of x109]     (x110) {$10$};
		\node[usym, right = of x110]     (x111) {$11$};
		\node[usym, right = of x111]     (x112) {$12$};
		\node[esym, right = of x112]     (x113) {$13$};
		\node[usym, right = of x113]     (x114) {$14$};
		\node[esym, right = of x114]     (x115) {$15$};
		\node      [right = of x115]     (x1end) {$\cdots$};		
	\end{tikzpicture}}
	\caption{An Example of Channel $\rm{II}$:  In any sliding window of length $W=5$ there is either a single erasure burst of length up to $B=3$ and possibly one isolated erasure, or $N=2$ isolated erasures. }
	\label{fig:chan-mix2}
	\end{minipage}%
\end{figure}
\end{center}
We consider a class of packet erasure channels where the erasure patterns are locally constrained. 
In any sliding window of length $W$, the channel can introduce only one of the following patterns:
\begin{itemize}
\item A single erasure burst of maximum length $B$ plus a maximum of $K$ isolated erasures or,
\item A maximum of $N$ erasures in arbitrary locations.
\end{itemize}
Note that ${N \le B+K}$. We denote such a channel  by $\cC(N,B,K,W)$.
We will focus on two special subclasses of such channels. 
The first class, Channel $\rm{I}$ is given by: $\cC_{\mrm{I}}(N,B,W) \defeq \cC(N,B,0,W),$ i.e., it only introduces
either a burst erasure or  up to $N$ arbitrary erasures. The second class, Channel $\rm{II}$ is given by:
${\cC_{\mrm{II}}(N,B,W) \defeq \cC(N,B,1,W)}$. It allows for one burst erasure of maximum length $B$ plus up to one isolated erasure,
or $N$ arbitrary erasures.  These specific channels are inspired by the dominant erasure events associated with GE channel as discussed in  section~\ref{sec:GE}. Fig.~\ref{fig:chan-mix} and~\ref{fig:chan-mix2} provide examples of channels $\cC_\mrm{I}(2,3,5)$ and $\cC_\mrm{II}(2,3,5)$ respectively. 

Clearly the erasure patterns associated with channel $\cC_\mrm{II}$ include
those associated with $\cC_\mrm{I}$. However we focus on channel $\cC_\mrm{I}$ first as it is simpler to analyze and reveals
some important insights. In particular for this class of channels we propose a near optimal class of codes based on a layered design. 
We first construct an optimal code for the burst erasure channel without considering the value of $N$. Thereafter we concatenate an additional layer of parity checks to recover from the isolated erasure patterns. Inspired by this result, for the channel $\cC_{\mrm{II}},$ we again construct a streaming code for the burst-erasure channel. A similar layering principle can be applied to treat isolated losses, although this extension is not discussed explicitly in the paper.

We next formally define a  {\em{streaming erasure code}}.
At each time ${i \ge 0}$, the encoder observes a source symbol $\bs[i]$, drawn
from a source alphabet\footnote{For keeping the notation compact, any symbol with a negative time-index will be assumed to be the constant zero symbol.} ${\cal S}$
and generates a channel symbol 
$\bx[i] = f_i(\bs[0], \ldots, \bs[i]) \in \cX$. 
The channel output is  either $\by[i] =\bx[i]$ or $\by[i] = \star,$ when the output is erased. The decoder is required
to reconstruct each packet with a delay of $T$ units i.e., for each $i\ge 0$ there exists a decoding function:
$\bs[i] = g_i(\by[0], \ldots, \by[i+T]).$ A rate $R = \frac{H(\bs)}{\log_2|\cX|}$ is achievable if there exists a feasible code that recovers every erased symbol $\bs[i]$ by time\footnote{Since we assume a deterministic channel model, it suffices to consider zero error probability. In our analysis we will assume that $\cS = {\mathbb F}_q^k$ and $\cX = {\mathbb F}_q^n$ i.e., the source and channel packets are vectors of lengths $k$ and $n$ respectively over some base field. We will assume that the field size, $q$ can be as large as required, and furthermore that $k$ and $n$ can also be as large as necessary with the rate given by $R = \frac{k}{n}$.} ${i+T}$.


\section{Preliminaries}
\label{prelim}
In this section we consider some previously studied  error correction codes and characterize their performance in our proposed setup. We discuss the special cases under which these codes are optimal. Our new constructions  use these codes as building blocks and therefore the review of these codes is useful.

\subsection{Strongly-MDS Codes}
\label{subsec:BEC}
Classical erasure codes are designed for maximizing the underlying distance properties.  Roughly speaking, such codes will recover all the missing source symbols simultaneously once sufficiently many parity checks have been received at the decoder. Indeed a commonly used family of such codes,  {\em random-linear codes}, see e.g.,~\cite{ho-med, roundrobin}, are designed to guarantee that the underlying system of equations is nearly of a full rank.
We discuss one particular class of {\em deterministic code} constructions with optimal distance properties  below.

Consider a $(n,k,m)$ convolutional code that maps an input source stream $\bs[i] \in {\mathbb F}_q^k$ to an output $\bx[i] \in {\mathbb F}_q^n$ using a memory $m$ encoder\footnote{We use $^\dagger$ to denote the vector/matrix transpose operation. In particular we will treat $\bs[i]$ and $\bx[j]$ as column vectors and therefore $\bs^\dagger[i]$ and $\bx^\dagger[j]$ denote the associated row vectors.}  i.e., 
\begin{align}
\bx[i] = \left(\sum_{t=0}^m \bs^\dagger[{i-t}] \cdot \bG_{t}\right)^\dagger, \label{eq:conv-code}
\end{align}
where $\bG_0,\ldots, \bG_m$ are ${k\times n}$  matrices with elements in ${\mathbb F}_q$.  
The first ${j+1}$ output symbols can be expressed as,
\begin{align}
\label{eq:trunc-cc}
\big[\bx^\dagger[0], \bx^\dagger[1], \ldots, \bx^\dagger[j]\big] = \big[\bs^\dagger[0], \bs^\dagger[1], \ldots, \bs^\dagger[j]\big] \cdot \bG_j^s.
\end{align}
where 
\begin{equation}\bG_j^s \defeq \begin{bmatrix}\bG_0 & \bG_1 & \ldots & \bG_j \\  0 & \bG_0 &  & \bG_{j-1} \\ \vdots & &\ddots & \vdots \\ 0 & & \ldots & \bG_0 \end{bmatrix} \in {\mathbb F}_q^{(j+1)k\times (j+1)n}\label{eq:GsT}\end{equation}
is the truncated generator matrix to the first ${j+1}$ columns and $\bG_j ={\bf 0}_{k\times n}$ if $j > m$. Furthermore the convolutional code is systematic if we can express each sub-generator matrix in the following form:
\begin{align}
\label{eq:systematic}
\bG_0=[\bI_{k\times k}~ {\bf 0}_{k \times n-k}], \qquad \bG_i = [{\bf 0}_{k \times k}~\bH_i],~i =1,\ldots, T
\end{align}
where $\bI_{k\times k}$ denotes the ${k\times k}$ identity matrix, ${\bf 0}$ denotes the zero matrix, and each ${\bH_i \in {\mathbb F}_q^{k \times (n-k)}}$.  For a systematic convolutional code, Eq.~\eqref{eq:conv-code} reduces to 
\begin{align}
\bx[i] = \left[\begin{array}{c}\bs[i] \\ \bp[i]\end{array}\right], \qquad \bp[i] = \left(\sum_{t=1}^m \bs^\dagger[i-t] \cdot\bH_t\right)^\dagger.  \label{eq:conv-code-sys}
\end{align}

We are particularly interested in a class of Strongly-MDS codes\footnote{In this paper we will only treat systematic Strongly-MDS Codes and for convenience drop the term ``systematic" when referring to these codes.}~\cite{strongly-mds} with the following properties.
\begin{lemma}
A $(n,k,m)$ (systematic) Strongly-MDS code has the following properties for each $j =0,1,\ldots,m$:
\begin{enumerate}
\item[P1.] Suppose that in the window $[0,j],$ there are no more than $(1-R)(j+1)$ erasures in arbitrary locations, then $\bs[0]$ is recovered by time $t=j$.
\item[P2.] Suppose an erasure burst happens in the interval $[0,B-1],$ where $B \le (1-R)(j+1)$, then all the symbols $\bs[0],\ldots, \bs[B-1]$ are simultaneously recovered at time $t=j$.
\end{enumerate}
\label{lem:mds}
\end{lemma}
\begin{proof}
See Appendix~\ref{app:MDS}.
\end{proof}

\vspace{1em}

The proof of Lemma~\ref{lem:mds} involves applying the properties of  Strongly-MDS codes proposed in~\cite{strongly-mds}. While the results in~\cite{strongly-mds} are stated for channels that can erase each sub-symbol in $\bx[i],$ they can be suitably adapted for our proposed setup. We delegate the proof to Appendix~\ref{app:MDS}.

As a direct consequence  of Lemma~\ref{lem:mds}, it can be seen that a rate $R$, Strongly-MDS code, can achieve any $N = B \le (1-R)(T+1)$ over the  $\cC_\mrm{I}(N,B,T+1)$ channel. As will be shown subsequently (c.f.~Theorem~\ref{thm:Chan-1-UB}) this is in fact the maximum value of $N$ that can be achieved. Nevertheless the largest value of $B$ can be higher. The class of codes treated in the following sub-section achieve the largest possible value of $B$.
\subsection{Maximally Short Codes}
\label{subsec:MS}
Maximally Short codes, introduced in~\cite{MartinianS04}, are streaming codes that correct the longest possible erasure burst in any sliding window of length ${T+1}$.  In particular the following result was established in~\cite{MartinianS04, MartinianT07}.
\begin{lemma}[Martinian and Sundberg~\cite{MartinianS04}]
\label{lem:ms}
Consider any channel that in any window of length ${T+1}$ introduces a single erasure burst of length no more\footnote{This statement is equivalent to saying that each erasure burst is of length no greater than $B$ and that successive bursts are separated by a distance of at-least $T$ time units. } than $B$. For any $(n,k,m)$ convolutional code which recovers every source packet $\bs[i]$ by time ${t=i+T}$ we must have that
\begin{align}
B \le T\cdot \min\left(1, \frac{1-R}{R}\right)\label{eq:B-UB}
\end{align}
Furthermore the upper bound in~\eqref{eq:B-UB} can be attained by the Maximally Short (MS) Codes.
\end{lemma}

The construction of MS codes presented in~\cite{MartinianS04,MartinianT07} involves first constructing a particular low-delay block code and then converting it into a convolutional code using a diagonal interleaving technique. This approach will not be used in this paper. Instead we discuss another construction~\cite{erlc-infocom} which is a slight generalization of the above approach.  In particular we show that for any $T$ and $B$ there exists a streaming code of rate $R = \frac{T}{T+B}$ that satisfies~\eqref{eq:B-UB} with equality. In the proposed construction we  split each source symbol $\bs[i] \in {\mathbb F}_q^{T}$ into two groups $\bu[i] \in {\mathbb F}_q^B$ and $\bv[i] \in {\mathbb F}_q^{T-B}$ as follows:
\begin{align}
\bs[i] = (\underbrace{u_0[i],\ldots, u_{B-1}[i]}_{=\bu[i]}, \underbrace{v_0[i],\ldots, v_{T-B-1}[i]}_{=\bv[i]})^\dagger.
\label{eq:s-split}
\end{align}
We apply a  $(T, T-B, T)$ Strongly-MDS code on the symbols $\bv[i]$ and generate parity check symbols
\begin{align}
\bp_v[i] = \left(\sum_{j=1}^{T} \bv^\dagger[i-j]\cdot \bH^v_j\right)^\dagger, \quad \bp_v[i] \in {\mathbb F}_q^{B}   \label{eq:pv-conc},
\end{align}
where the matrices $\bH_j^v$  are $(T-B) \times B$ matrices  associated with the Strongly-MDS code~\eqref{eq:systematic}.
Superimpose the $\bu[\cdot]$ symbols onto $\bp_v[\cdot]$ and let \begin{align}
\bq[i] = \bp_v[i] + \bu[i-T]. \label{eq:MS-q}
\end{align}The channel input at time $i$ is given by $\bx[i] = \left(\bu^\dagger[i],\bv^\dagger[i],\bq^\dagger[i]\right)^\dagger \in {\mathbb F}_q^{T+B}$. 

For decoding, we suppose an erasure burst of length $B$ spans the interval $[0,B-1]$. The decoder starts by recovering the parity check symbols $\bp_v[B],\dots,\bp_v[T-1]$ by subtracting the unerased symbols $\bu[B-T],\dots,\bu[-1]$, respectively. These recovered parity checks are used to recover the erased symbols $\bv[0],\dots,\bv[B-1]$ using property P2 in Lemma~\ref{lem:mds}. Now, the parity check symbols $\bp_v[j]$ for $j \ge T$ can be computed as they combine $\bv[\cdot]$ symbols which are either not erased or recovered in the previous step. Hence, these parity checks can be subtracted from the corresponding $\bq[\cdot]$ to recover $\bu[0],\dots,\bu[B-1]$ at times $T,\dots,T+B$, respectively, i.e., with a delay of $T$ symbols and the decoding is complete.

Unfortunately the MS Codes are not robust against isolated erasures. In fact it can be easily seen that for the $\cC_\mrm{I}(N, B, T+1)$
channel they only attain $N=1$.

\section{Channel $\rm{I}$: Code Constructions and Bounds}
\label{sec:chan-I}

Throughout the analysis of Channel $\rm{I}$, we select ${W=T+1}$ where recall that $T$ is the decoding delay.
Note that every source symbol $\bs[i]$ remains active for a duration of ${T+1}$ symbols before its deadline expires. Therefore the erasure patterns observed in a window of length ${T+1}$ are naturally of interest. For such channels we will suppress the $W$ parameter in the notation for $\cC_\mrm{I}(\cdot)$ and simply use the notation $\cC_\mrm{I}(N,B)$.
The extensions to other values of $W$ will be reported elsewhere.

In this section we propose a new family of streaming codes that are near optimal for all rates for the channel $\cC_\mrm{I}(N,B)$. Before stating our constructions, we use the following upper bound~\cite{erlc-infocom}.
\begin{thm}[Badr et al.\ \cite{erlc-infocom}]
Any achievable rate  for  $\cC_I(N,B),$ satisfies
\begin{align}
\left(\frac{R}{1-R}\right)B + N \le T+1. \label{eq:r-ub}
\end{align}
\label{thm:Chan-1-UB}
and furthermore $N \le B$ and $B \le T$. 
\hfill$\Box$
\end{thm}
Theorem~\ref{thm:Chan-1-UB} shows that when the rate $R$ and delay $T$ are fixed  there exists a tradeoff between the achievable values of $B$ and $N$. We cannot have streaming codes that simultaneously correct long erasure bursts and many isolated erasures. The upper bound~\eqref{eq:r-ub}
 also shows that the $R=1/2$ codes found via a computer search in~\cite[Section V-B]{MartinianS04} are indeed optimal.

We propose a class of codes, {\em Maximum Distance And Span Tradeoff (MiDAS) codes}, that achieve near-optimal tradeoff.

\begin{thm}
For the  channel $\cC_I(N,B)$ and delay $T$ with $N \le B$ and $B \le T,$ there exists a code of rate $R$ that satisfies
\begin{align}
\left(\frac{R}{1-R}\right)B + N \ge T.
\label{b-achiev}
\end{align}
$\hfill\Box$
\label{thm:midas}
\end{thm}

\begin{remark}
Comparing~\eqref{b-achiev} with the upper bound in~\eqref{eq:r-ub}, the only difference is the additional constant $1$ in the right hand side of the inequalities.
\end{remark}

Our proposed construction is illustrated in Fig.~\ref{fig:midas-const}.  We split the source symbol $\bs[i] \in {\mathbb F}_q^{T}$ into two groups $\bu[i]$ and $\bv[i]$ as in~\eqref{eq:s-split} and generate the parity checks $\bq[i]$ as in~\eqref{eq:MS-q}. The resulting code up to this point is just a MS code that can correct an erasure burst of length $B$. 

We further apply another  Strongly-MDS code to the $\bu[i]$ symbols and generate the set of parity check symbols,\begin{align}
\bp_u[i] = \left(\sum_{j=1}^{T} \bu^\dagger[i-j]\cdot \bH^u_j\right)^\dagger,  \qquad \bp_u[i] \in {\mathbb F}_q^K\label{eq:pu-conc},
\end{align}
where $\bH^u_j$ are ${B \times K}$ matrices associated with a Strongly-MDS code~\eqref{eq:systematic}.
We finally concatenate the parity checks $\bq[i]$ and $\bp_u[i]$ with the source symbols i.e., ${\bx[i] = \left(\bu^\dagger[i], \bv^\dagger[i], \bq^\dagger[i], \bp_u^\dagger[i]\right)^\dagger}$ as shown in Fig.~\ref{fig:midas-const}.

\begin{figure}[t]
\resizebox{\columnwidth}{!}{
\begin{tikzpicture}[node distance=1mm,
  ]
  \tikzstyle{triple2} = [rectangle split, anchor=text,rectangle split parts=4]
  \tikzstyle{double2} = [rectangle split, anchor=text,rectangle split parts=2]
  \tikzstyle{triple} = [draw, rectangle split,rectangle split parts=4]
	\tikzstyle{double} = [draw, rectangle split,rectangle split parts=2]
	\tikzset{block/.style={rectangle,draw}}
	
	
	\node[triple2, minimum width=1.6cm] (start) {$B$ Symbols
    \nodepart{second}
      $T-B$ Symbols
    \nodepart{third}
     \tikz{\node[double2] {$B$ \nodepart{second}Symbols};}
		\nodepart{fourth}
		$K$ Symbols
  };
	
  \node[triple,  right = of start,fill=light-gray,minimum width=1.6cm] (p1) {$\bu[0]$
    \nodepart{second}
      $\bv[0]$
    \nodepart{third}
      \tikz{\node[double2] {$\bu[-T]$ \nodepart{second}$+\bp_v(\bv^{-1})$};}
			\nodepart{fourth}
			$\bp_u(\bu^{-1})$
  };
  
  \node[triple, right = of p1,fill=light-gray,minimum width=1.6cm] (p11) {$\bu[1]$
    \nodepart{second}
      $\bv[1]$
    \nodepart{third}
      \tikz{\node[double2] {$\bu[-T+1]$ \nodepart{second}$+\bp_v(\bv^{0})$};}
			\nodepart{fourth}
			$\bp_u(\bu^{0})$
  };
  
  \node [block, right of=p11,minimum width=1.65cm, minimum height=3.28cm,node distance=1.94cm, fill=light-gray] (p2) {$\cdots$};
   
  \node[triple,  right = of p2,fill=light-gray,minimum width=1.6cm] (p3) {$\bu[B-1]$
    \nodepart{second}
      $\bv[B-1]$
    \nodepart{third}
      \tikz{\node[double2] {$\bu[-T+B-1]$ \nodepart{second}$+\bp_v(\bv^{B-2})$};}
			\nodepart{fourth}
			$\bp_u(\bu^{B-2})$
  };
  
  \node[triple,  right = of p3,minimum width=1.6cm] (p4) {$\bu[B]$
    \nodepart{second}
      $\bv[B]$
    \nodepart{third}
      \tikz{\node[double2] {$\bu[-T+B]$ \nodepart{second}$+\bp_v(\bv^{B-1})$};}
			\nodepart{fourth}
			$\bp_u(\bu^{B-1})$
  };
  
  \node [block, right of=p4,minimum width=1.65cm, minimum height=3.28cm,node distance=1.96cm] (p5) {$\cdots$};
  
  \node[triple,  right = of p5,minimum width=1.6cm] (p6) {$\bu[T-1]$
    \nodepart{second}
      $\bv[T-1]$
    \nodepart{third}
      \tikz{\node[double2] {$\bu[-1]$ \nodepart{second}$+\bp_v(\bv^{T-2})$};}
			\nodepart{fourth}
			$\bp_u(\bu^{T-2})$
  };
  
  \node[triple, right = of p6,minimum width=1.6cm] (p7) {$\bu[T]$
    \nodepart{second}
      $\bv[T]$
    \nodepart{third}
      \tikz{\node[double2] {$\bu[0]+$ \nodepart{second}$\bp_v(\bv^{T-1})$};}
			\nodepart{fourth}
			$\bp_u(\bu^{T-1})$
  };
  
  \bracedn{p1}{p3}{-2mm}{\footnotesize{Erased Packets}};
  \bracedn{p4}{p6}{-2mm}{\footnotesize{Used to recover $\bv[0],\cdots ,\bv[B-1]$}};
  \bracedn{p7}{p7}{-2mm}{\footnotesize{Recover $\bu[0]$}};

\end{tikzpicture}}
\caption{A window of $T+1$ channel packets showing the decoding steps of a MiDAS code when an erasure burst takes place. In the case of isolated erasures, the $\bv$ and $\bu$ symbols are recovered separately using the $\bp_v(\cdot)$ parities in the window $[0,T-1]$ and $\bp_u(\cdot)$ parities in the window $[0,T]$, respectively. $\bv^{t}$ denotes the set of symbols $(\bv[t-T+1],\dots,\bv[t])$.}
\label{fig:midas-const}
\end{figure}

Note that $\bx[i] \in {\mathbb F}_q^{(T+B+K)}$ and the associated rate is given by  $R = \frac{T}{T+B+K}$.
In our setup $K$ is a free parameter which can be used to vary $N$. The resulting value of $N$ is computed next.


Our proposed decoder  recovers $\bv[0]$ from  the parity checks $\bp_v[i]$ in the interval $i \in [0,T-1]$. It separately uses the parity checks $\bp_u[i]$ in $i \in [0,T]$ to recover $\bu[0]$. Recall from~\eqref{eq:MS-q}  that
$\bq[i] = \bp_v[i] + \bu[i-T],$ where $\bp_v[i]$ are the parity checks of the Strongly-MDS code~\eqref{eq:pv-conc}.
Since the interfering $\bu[\cdot]$ symbols in this interval are not erased, they can be canceled out 
by the decoder from $\bq[\cdot]$ and the corresponding parity checks $\bp_v[\cdot]$ are recovered at the decoder. Since the code $(\bv[i], \bp_v[i])$ is a Strongly-MDS code of rate $R_v = \frac{T-B}{T}$, applying Lemma~\ref{lem:mds}  the associated value of $N^v$ is given by $N^v = (1-R_v)T= B$.   Since $N \le B$ holds, the recovery of $\bv[0]$ by time ${t=T-1}$ is guaranteed by the code construction. 

As stated before, for recovering $\bu[0]$ at time $t=T,$ we use the $\bp_u[\cdot]$ parity checks. Note that the associated code
$(\bu[i], \bp_u[i])$ is a Strongly-MDS code with rate $R_u = \frac{B}{B+K}$ and hence it follows from Lemma~\ref{lem:mds} that in order to achieve the desired value of $N$ we must have that
\begin{align}
N = (1-R_u)(T+1) = \frac{K}{K+B}(T+1)
\end{align}
Thus it suffices to select\footnote{If the right hand side in~\eqref{eq:K-opt} is not an integer we will need to expand each of the source symbol by a factor of $(T+1-N)$ i.e., $\bs[i] \in {\mathbb F}^{T(T+1-N)}$ and similarly $\bu[i]\in {\mathbb F}^{B(T+1-N)}$ and $\bv[i]\in {\mathbb F}^{(T-B)(T+1-N)}$ in~\eqref{eq:s-split}  also need to be expanded by the factor $(T+1-N)$.} 
\begin{align}
K = \frac{N}{T+1-N}B,\label{eq:K-opt}
\end{align}
and the rate of the code satisfies
\begin{align}
R &= \frac{T}{T+B+K} = \frac{T}{T+B + B\frac{N}{T+1-N}}\label{eq:K-sub}\\
&\ge \frac{T}{T+B + B\frac{N}{T-N}} \notag\\
&= \frac{T-N}{T-N+B}\label{eq:R-lb-final}
\end{align}
where~\eqref{eq:K-sub} follows by substituting in~\eqref{eq:K-opt}. Rearranging~\eqref{eq:R-lb-final}
we have that
\begin{align}
\frac{R}{1-R}B + N \ge T.\label{eq:midas-lb}
\end{align}
We have thus far shown that if there are no more than $N$ erasures in the interval $[0,T]$ then $\bs[0]$ can be recovered by time $T$. Once $\bs[0]$ is recovered we can cancel its effect from all the available parity checks. Thus the same argument can be used to show that $\bs[1]$ can be recovered at time ${T+1}$ if there are no more than $N$ erasures in the interval $[1,T+1]$.  Recursively continuing this argument we are guaranteed the recovery of each $\bs[i]$ by time ${i+T}$. The proof of Theorem~\ref{thm:midas} is thus completed.

$\hfill\blacksquare$

\subsection{Example}
Table~\ref{tab:MIDAS_N2B3T7} illustrates a MiDAS code of rate $\frac{7}{11}$ designed for $\cC_{\mrm{I}}(2,3)$ with $T=7$.  Note that from, we split each source symbol $\bs_i$ into $\bu_i$ and $\bv_i$ as in~\eqref{eq:s-split}. We have that $\bu_i\in {\mathbb F}_q^3$ and $\bv_i \in {\mathbb F}_q^4$ and from~\eqref{eq:pv-conc} and~\eqref{eq:K-opt} we have that the parity check symbols $\bp_v(\cdot) \in {\mathbb F}_q^{3}$ and $\bp_u(\cdot) \in {\mathbb F}_q$. For convenience, we use the subscript of $\bu$ and $\bv$  denotes the time index, whereas the notation $\bv^t$ denotes a sequence of $T$ consecutive symbols ${\bv_t, \bv_{t-1},\ldots, \bv_{t-T+1}}$. Thus $\bp_v(\bv^t)$ denotes the parity checks ~\eqref{eq:pv-conc} applied to $\bv^t$ and similarly $\bp_u(\bu^t)$ denotes the  parity checks applied to $\bu^t$. Let $\bq_i = \bp_v(\bv^{i-1}) + \bu_{i-T}$. The input at time $i$, $\bx_i$, corresponds the the $i-$th column in  Table~\ref{tab:MIDAS_N2B3T7}.

\begin{table}[!htb]
	\centering
	\begin{align*}
		\begin{array}{c|c|c|c|c|c|c|c|c|c|c}
			 |.|& 0 			 				&  1  					&  2  					&  3  					&  4  					&  5  					&  6  					&  7  					& 8  						&  9 \\ \hline
			 3 	& \bu_0  					&  \bu_1  			&  \bu_2  			&  \bu_3  			&  \bu_4  			&  \bu_5  			&  \bu_6  			&  \bu_7  			& \bu_8  				& \bu_9 \\
			 4 	& \bv_0  					&  \bv_1  			&  \bv_2  			&  \bv_3  			&  \bv_4  			&  \bv_5  			&  \bv_6  			&  \bv_7  			& \bv_8  				& \bv_9 \\
			 3 	& \bp_v(\bv^{-1}) & \bp_v(\bv^0) 	& \bp_v(\bv^1) 	& \bp_v(\bv^2) 	& \bp_v(\bv^3) 	& \bp_v(\bv^4) 	& \bp_v(\bv^5) 	& \bp_v(\bv^6) 	& \bp_v(\bv^7)	& \bp_v(\bv^8) \\
			 		& +\bu_{-7} 			& +\bu_{-6} 		& +\bu_{-5} 		& +\bu_{-4} 		& +\bu_{-3} 		& +\bu_{-2} 		& +\bu_{-1} 		& +\bu_{0} 		& +\bu_{1}  		& +\bu_{2} \\
			 1 	& \bp_u(\bu^{-1}) & \bp_u(\bu^0) 	& \bp_u(\bu^1) 	& \bp_u(\bu^2) 	& \bp_u(\bu^3) 	& \bp_u(\bu^4) 	& \bp_u(\bu^5) 	& \bp_u(\bu^6) 	& \bp_u(\bu^7) 	& \bp_u(\bu^8) 
		\end{array}
		\end{align*}
	\caption{{\scriptsize{A MiDAS Code with $(N,B)=(2,3)$ and rate $\frac{7}{11}$ for a delay of $T=7$. The number of sub-symbols in each group is denoted in the first column.}}}
	\label{tab:MIDAS_N2B3T7}
	\vspace{-1em}
\end{table}
To illustrate the decoding, assume an erasure burst of length $B=3$ happens in the interval $[0,2]$. In Table~\ref{tab:MIDAS_N2B3T7}, $\bu_{-4},\dots,\bu_{-1}$ are not erased and thus can be subtracted to recover the parity checks $\bp_v$ in the interval $[3,6]$. These parities belong to the $(7,4,7)$ Strongly-MDS Code and suffice to recover $\bv_0,\ldots, \bv_2$ simultaneously at time $T-1=6$. In particular note that  $(1-R_v)T = (1-\frac{4}{7})(7) = 3 $ and thus property P2 of Lemma~\ref{lem:mds} can be immediately applied. 
At this point all the erased $\bv[\cdot]$ symbols have been recovered. Next, $\bu_0$ is recovered at time $7$ by subtracting out $\bp_v(\bv^8)$. Likewise $\bu_1$ and $\bu_2$ are sequentially recovered at time $8$ and $9$. 

To compute the achievable $N$, it suffices to compute $N^v$ for the $\left(\bv, \bp_v(\cdot)\right)^\dagger$ code in the interval $[0,6]$ and $N^u$ for the $\left(\bu, \bp_u(\cdot)\right)^\dagger$ code in the interval $[0,7]$. Using that $R_v = 4/7$ and $R_u = 3/4$ it follows from Lemma~\ref{lem:mds} that $N^v=(1-R_v)(T)=3$ and similarly $N^u=(1-R_u)(T+1) = 2$ and hence $N=\min(N^v, N^u)=2$.

\subsection{Numerical Comparison}
Table~\ref{Sco:Codes} summarizes the achievable $N$ and $B$ for different codes\footnote{We note that the floor of the values given in Table~\ref{Sco:Codes} should be considered as the values might not be integers}. The first three rows correspond to Strongly-MDS, MS and MiDAS codes discussed in Section~\ref{subsec:BEC},~\ref{subsec:MS} and~\ref{sec:chan-I}. The last row correspond to a family of codes -- Embedded Random Linear Codes (E-RLC) -- propsed by Badr et al.~\cite{erlc-infocom} for the channel $\cC_\mrm{I}(N,B)$ and delay $T$ whose rate satisfies
\begin{align}
\left(\frac{R}{1-R}\right)(B+N-1) \ge T.
\label{b-slec}
\end{align}

While such constructions are optimal for $R=1/2,$ they are sub-optimal in general.

\begin{table}[!htdp]
\begin{center}
\begin{tabular}{c|c|c}
Code & $N$ & $ B$ \\\hline\hline
Strongly-MDS Codes & $(1-R)(T+1)$ &$(1-R)(T+1)$ \\\hline
Maximally Short Codes & $1$ & $T\cdot\min\left(\frac{1}{R}-1, 1\right)$ \\\hline
 \!\!MiDAS Codes  \!\! & $\min\left(B, T\!- \!\frac{R}{1-R}B \right) $& $B \in [1, T]$  \\\hline
E-RLC Codes~\cite{erlc-infocom} & & \\
 ${\Delta \in [R(T+1), T-1]},$ &$\frac{1-R}{R}(T-\Delta)$ &  $\frac{1-R}{R}\Delta$  \\
 ${R \ge 1/2}$ & & \\\hline
\end{tabular}
\end{center}
\caption{Achievable $(N,B)$ for channel $\cC_\mrm{I}(\cdot)$}
\label{Sco:Codes}
\end{table}%

\begin{figure}
\centering
\includegraphics[width=\linewidth]{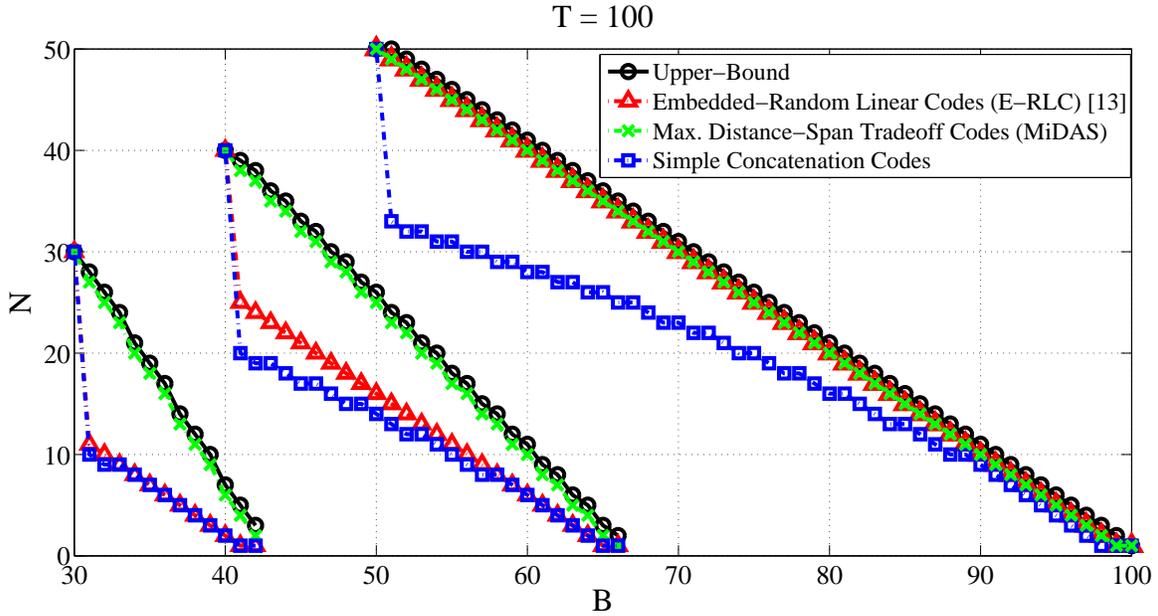}
\caption{Achievable tradeoff between $N$ and $B$ for different rates $R \in \{0.5, 0.6, 0.7\}$ from right to left. The uppermost curve  (solid black lines with `o') in each group is the upper bound in~\eqref{eq:r-ub}. The MiDAS codes 
 are shown with broken green lines with `$\times$' and are very close to the upper bound. The  Simple Concatenation codes are shown with broken blue lines with `$\square$' whereas the E-RLC codes in~\cite{erlc-infocom} are shown with broken red lines with `$\triangle$'. The delay is fixed to $T=100$. }
\label{fig:cTdT_Tradeoff}
\end{figure}

We further numerically illustrate the achievable $(N,B)$ pairs for various codes in Fig.~\ref{fig:cTdT_Tradeoff}. We note that the  Strongly-MDS and MS Codes in sections~\ref{subsec:BEC} and~\ref{subsec:MS} respectively only achieve the extreme points on the tradeoff. The MiDAS codes achieve a tradeoff, very close to the upper bound for all rates. The E-RLC codes achieving the bound in~\eqref{b-slec} are generally suboptimal except for $R=0.5$. Fig.~\ref{fig:cTdT_Tradeoff} also illustrates another class of codes based on a naive concatenation approach of MS and Strongly-MDS codes. These are always suboptimal and hence not discussed.

\section{Partial Recovery Codes for Channel $\rm{II}$}
\label{sec:chan-2}

In this section we study streaming erasure codes for channel $\cC_\mrm{II}(N,B,W)$. Recall that in any sliding window of length $W$, such channels permit erasure patterns consisting of (i) one erasure burst plus one isolated erasure either before or after the burst, or (ii) up to $N$  erasures in arbitrary locations. The burst plus isolated erasure pattern captures the transition between the bad and good states on the GE channel. Such events appear dominant when the delay $T$   
is large, as discussed in section~\ref{sec:GE}.  The MiDAS codes are not effective when such patterns are dominant.

Motivated by the MiDAS code construction we  focus on codes that correct one erasure burst plus one isolated erasure.
 One can extend the construction to achieve any desired $N,$ by suitably concatenating additional parity checks. However we do not discuss this extension in the paper. In the rest of this section we will drop the parameter $N$ and refer to such channels as $\cC_\mrm{II}(B,W)$.

\begin{figure}[htbp]
	\centering
	\resizebox{5cm}{!}{
	\begin{tikzpicture}[node distance=0mm]
		\node                       (x1start) {Link:};
		\node[usym, right = of x1start]  (x100) {};
		\node[usym, right = of x100]     (x101) {};
		\node[usym, right = of x101]     (x102) {};
		\node[usym, right = of x102]     (x103) {};
		\node[usym, right = of x103]     (x104) {};
		\node[esym, right = of x104]     (x105) {};
		\node[esym, right = of x105]     (x106) {};
		\node[esym, right = of x106]     (x107) {};
		\node[esym, right = of x107]     (x108) {};
		\node[usym, right = of x108]     (x109) {};
		\node[usym, right = of x109]     (x110) {};
		\node[usym, right = of x110]     (x111) {};
		\node[usym, right = of x111]     (x112) {};
		\node[esym, right = of x112]     (x113) {};
		\node      [right = of x113]     (x1end) {$\cdots$};

		\dimdn{x100}{x104}{-2mm}{$T$};
		\dimdn{x105}{x108}{-2mm}{$B$};
		\dimdn{x109}{x112}{-2mm}{$T-1$};
		\dimup{x100}{x113}{4mm}{$2T + B$};
	\end{tikzpicture}}
	\caption{An isolated erasure associated with a burst in channel $\cC_\mrm{II}$}
	\label{fig:period}
\end{figure}
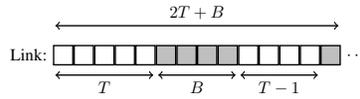

\begin{defn}
\label{def:associated}
An isolated erasure is defined to be \emph{associated} with the erasure burst, if it occurs within the $T$ symbols before or after this burst.
\end{defn}

Throughout this section, we let ${W=2T+B}$ since we are interested in an interval of length $T$ before or after the erasure burst. Note that every erasure burst has at most one associated isolated erasure (c.f. Fig.~\ref{fig:period}). Conversely every isolated erasure can be associated with no more than one erasure burst. 

It turns out that codes that achieve perfect recovery over $\cC_\mrm{II}$ require a significant overhead, particularly when $T$ is close to $B$. Therefore we consider partial recovery codes as discussed next.

\begin{defn}
\label{def:partial}
A \emph{Partial Recovery Code (PRC)}  for $C_\mrm{II}(B,W=2T+B)$ recovers all but one source symbol with a delay of $T$ in each pattern consisting of an erasure burst and its associated isolated erasure.
\end{defn}

\begin{thm}
\label{thm:PRC}
There exists a partial recovery code for $\cC_\mrm{II}(B,2T+B)$ of rate,
\begin{align}
\label{eq:rate-PRC}
R = \max_{B<\Delta<T} \frac{\Delta(T-\Delta) + (B+1)}{\Delta(T-\Delta) + (B+1)(T-\Delta+2)},
\end{align}
that satisfy Definition~\ref{def:partial}.
\end{thm}

\subsection{Code Construction}
\label{sec:PRC-construction}
The main steps in our proposed construction of a partial recovery code for $\cC_\mrm{II}(B, 2T+B)$ are as follows. Let $u$, $v$, $s$ and $\Delta$ be  integers that will be specified in the sequel. We let $\bs[i] \in {\mathbb F}_q^{u+v}$.
\begin{enumerate}
\item {\bf Source Splitting}: As in~\eqref{eq:s-split} we split $\bs[i]$ into two groups $\bu[i] \in {\mathbb F}_q^u$ and $\bv[i] \in {\mathbb F}_q^v$.
\item {\bf Construction of $\cC_{12}$}: We apply a rate $R_{12} = \frac{v}{v+u+s}$ Strongly-MDS  $(v+u+s, v, T)$ code $\cC_{12}: (\bv[i], \bp[i])$ to the $\bv[\cdot]$ symbols to generate parity check symbols $\bp[\cdot] \in {\mathbb F}_q^{u+s}$,
\begin{align}
\label{eq:pc-def}
\bp[i] = \left(\sum_{j=1}^T \bv^\dagger[i-j]\cdot \bH_j \right)^\dagger,
\end{align}
where $\bH_1,\ldots, \bH_T \in {\mathbb F}_q^{v \times (u+s)}$  are the matrices associated with the Strongly-MDS code~\eqref{eq:conv-code-sys}.
\item{\bf Parity Check Splitting}: We split each $\bp[i]$ into two groups $\bp_1[i] \in {\mathbb F}_q^u$ and $\bp_2[i] \in {\mathbb F}_q^s$ by assigning the first $u$ sub-symbols in $\bp[i]$ to $\bp_1[i]$ and the remaining $s$ sub-symbols of $\bp[i]$ to $\bp_2[i]$. We can express:
\begin{align}
\bp_k^\dagger[i] = \sum_{j=1}^{T} \bv^\dagger[i-j] \cdot \bH_j^k, \quad k=1,2
\end{align}
where the matrices $\bH_j^k$ are given by $\bH_j = [\bH_j^1~|~\bH_j^2]$. It can be shown that both $\bH_j^1$ 
and $\bH_j^2$ satisfy the Strongly-MDS property~\cite[Theorem 2.4]{strongly-mds} and therefore the codes $\cC_1: (\bv[i], \bp_1[i])$ and $\cC_2: (\bv[i], \bp_2[i])$ are both Strongly-MDS codes.

\item {\bf Repetition Code}: We combine a shifted copy of $\bu[\cdot]$ with the $\bp_1[\cdot]$ parity checks to generate $\bq[i] = \bp_1[i] + \bu[i-\Delta]$. Here $\Delta \in \{B+1,\dots,T\}$ denotes the shift applied to the $\bu[\cdot]$ stream before embedding it onto the $\bp_1[\cdot]$ stream.
\item {\bf Channel Symbol}: We concatenate the generated layers of parity check symbols to the source symbol to construct the channel symbol,
\begin{align}
\label{eq:x-PRC}
\bx[i] = (\bs^\dagger[i],\bq^\dagger[i],\bp_2^\dagger[i])^\dagger.
\end{align}
\end{enumerate}

The rate of the code in~\eqref{eq:x-PRC} is clearly $R = \frac{u+v}{2u+v+s}$. We further select the codes $\cC_{12}$ and $\cC_2$ to have the following rates:
\begin{align}
R_{12} &= \frac{v}{v+u+s} = \frac{\Delta-B-1}{\Delta}, \label{eq:r12-PRC} \\
R_{2} &= \frac{v}{v+s} = \frac{T-\Delta+1}{T-\Delta+2}. \label{eq:r2-PRC}
\end{align}
As established in Appendix~\ref{app:PRC}, these parameters guarantee that the above construction is a PRC code for $\cC_\mrm{II}(B,W=2T+B)$.

We further use the following values of $u$, $v$ and $s$ that satisfy~\eqref{eq:r12-PRC} and~\eqref{eq:r2-PRC},
\begin{align}
\label{eq:PRC_uvs}
u &= (B+1)(T-\Delta+1)-(\Delta-B-1) \nonumber \\
v &= (T-\Delta+1)(\Delta-B-1) \nonumber \\
s &= \Delta-B-1,
\end{align}
and the corresponding rate of the PRC code is
\begin{align}
R = \max_{B<\Delta \le T} \frac{\Delta(T-\Delta) + (B+1)}{\Delta(T-\Delta) + (B+1)(T-\Delta+2)},
\end{align}
which meets the rate given in Theorem~\ref{thm:PRC}.  This completes the details of the code construction. 
The decoding steps are described in detail in Appendix~\ref{app:PRC}.

\begin{remark}
If we assume that the source alphabet is sufficiently large such that the integer constraints can be ignored then the optimal shift is given by
\begin{align}
\label{eq:PRC_Delta}
\Delta^{\star} = \arg \max_{\Delta} R(\Delta) = T+1-\sqrt{T-B},
\end{align}
and the corresponding rate is given by
\begin{align}
R^{\star}  = \frac{(T+2)\sqrt{T-B}  - 2(T-B)}{(T+B+3)\sqrt{T-B}  - 2(T-B)}.
\label{eq:part-rec-opt}
\end{align}
\end{remark}

\subsection{Example}
We illustrate a PRC construction for $\cC_\mrm{II}(B=3, W=17)$ with $T=7$ and rate $R=5/9$.  The values of the code parameters are $u=6$, $v=4$, $s=2$ and $\Delta=6$.

\begin{table}[!htb]
	\centering
	\begin{align*}
		\begin{array}{c|c|c|c|c|c|c|c|c|c|c}
			 |.|& 0 			 				&  1  					&  2  					&  3  					&  4  					& 5 						& 6  						&  7  					& 8 						& 9  						\\ \hline
			 6 	& \bu_0  					&  \bu_1  			&  \bu_2  			&  \bu_3  			&  \bu_4  			& \bu_5  				& \bu_6					&  \bu_7  			& \bu_8  				& \bu_9  				\\
			 4 	& \bv_0  					&  \bv_1  			&  \bv_2  			&  \bv_3  			&  \bv_4  			& \bv_5  				& \bv_6  				&  \bv_7  			& \bv_8  				& \bv_9  				\\
			 6 	& \bp_1(\bv^{-1}) & \bp_1(\bv^0) 	& \bp_1(\bv^1) 	& \bp_1(\bv^2) 	& \bp_1(\bv^3) 	& \bp_1(\bv^4)	& \bp_1(\bv^5)	& \bp_1(\bv^6) 	& \bp_1(\bv^7)	& \bp_1(\bv^8)	\\
	 				& +\bu_{-6} 			& +\bu_{-5} 		& +\bu_{-4} 		& +\bu_{-3} 		& +\bu_{-2} 		& +\bu_{-1} 		& +\bu_{0}  		&+\bu_{1}				& +\bu_{2} 			& +\bu_{3}  		\\
			 2 	& \bp_2(\bv^{-1}) & \bp_2(\bv^0) 	& \bp_2(\bv^1) 	& \bp_2(\bv^2) 	& \bp_2(\bv^3) 	& \bp_2(\bv^4)	& \bp_2(\bv^5)  & \bp_2(\bv^6) 	& \bp_2(\bv^7)	& \bp_2(\bv^8)	  
		\end{array}
		\end{align*}
	\caption{A Partial Recovery Code with $B=3$ achieving a rate of $\frac{5}{9}$ for a delay of $T=7$.}
	\label{tab:Partial_B3T7}
	\vspace{-3em}	
\end{table}

\subsubsection*{Encoding}
\begin{enumerate}
\item Assume that each source symbol $\bs_i \in {\mathbb F}_q^{10}$ consists of ten sub-symbols and split it into $\bu_i$ and $\bv_i$, consisting of six and four sub-symbols respectively. 
\item Apply a $(u+v+s, v)= (12,4)$ Strongly-MDS code $\cC_{12}$ to the $\bv_i$ symbols to generate parity checks $\bp(\bv^{i-1}) \in {\mathbb F}_q^8$. Note that the rate of $\cC_{12}$ equals $R_{12} = \frac{1}{3}$.

\item We split each $\bp(\bv^{i-1})$ into $\bp_1(\bv^{i-1})$ and $\bp_2(\bv^{i-1})$ consisting of $6$ and $2$ sub-symbols as  shown in Table~\ref{tab:Partial_B3T7}. Note that the codes $\cC_1 : (\bv_i, \bp_1(\bv^{i-1}))$ and $\cC_2: (\bv_i, \bp_2(\bv^{i-1}))$ are both Strongly-MDS codes of rates $R_1 = 2/5$ and $R_2 = 2/3$ respectively.

\item We generate  $\bq_i = \bp_1(\bv^{i-1}) + \bu_{i-\Delta}$ and let the transmit symbol be $$\bx_i = \left(\bu_i^\dagger, \bv_i^\dagger, \bq_i^\dagger, \bp_2^\dagger(\bv^i)\right)^\dagger,$$ which corresponds to one column in Table~\ref{tab:Partial_B3T7}.
The overall rate equals $R=5/9$.
\end{enumerate}

\subsubsection*{Decoding}

We start by considering the case when the burst erasure precedes the isolated erasure.  Without loss of generality suppose that the burst erasure spans $t \in [0,2]$ and the isolated erasure occurs at time $t > 2$. 
\begin{itemize}
\item $t\in[3,4]:$ We use code $\cC_{12}$ in the interval $[0, \Delta-1] = [0,5] $ to recover the erased symbols $\bv_0,\bv_1, \bv_2, \bv_t$ by time $\tau=5$. Note that in this  interval any interfering $\bu$ symbols are not erased and can be cancelled out from $\bq_i$. Using $\tau=\Delta-1$ we have
\begin{align}
(1-R_{12})(\tau+1) = \left(1-\frac{1}{3}\right) 6 =4. 
\end{align}
Thus by using an extension of Lemma~\ref{lem:mds} (see Appendix~\ref{app:PRC}, Lemma~\ref{lem:twobursts}) the recovery of $\bv_0, \bv_1, \bv_2, \bv_t$ follows.

Thus by time $\tau=5,$ all the erased $\bv_i$ symbols have been recovered. The symbols $\bu_0,\bu_1,\bu_2$ and $\bu_t$ are each recovered at time $\tau = 6, 7, 8$ and $t +\Delta$ respectively by cancelling $\bp_1(\bv^{\tau-1})$ from the corresponding parity checks in Table~\ref{tab:Partial_B3T7}. Thus all the erased symbols get recovered by the decoder for these erasure patterns.

\item $t \ge 5:$ We use code $\cC_{12}$ in the interval  $[0,4]$ to recover the erased symbols  $\bv_0,\bv_1,\bv_2$ at time $\tau=4$. This follows by applying property P2 in Lemma~\ref{lem:mds} since
\begin{align}
(1-R_{12})(\tau+1) = \left(1-\frac{1}{3}\right) 5  > 3. 
\end{align}

Once the erasure burst has been recovered, the erased symbol $\bv_t$ can be recovered at time $\tau={t+T-\Delta+1} =t+2$ using the parity checks $\bp_2(\bv^{t-1})$ associated with $\cC_2$ in the interval $[t,t+2]$.  This again follows from Lemma~\ref{lem:mds} since
\begin{align}
(1-R_2)(T-\Delta+2) = \left(1-\frac{2}{3}\right)3 = 1
\end{align}
which suffices to recover the missing $\bv_t$.

For the recovery of $\bu_t$, note that by time $t+\Delta$ all the erased $\bv_i$ symbols have been recovered and hence the associated partiy check $\bp_1(\bv^{t+\Delta-1})$  can be cancelled from $\bq_{t+\Delta}$ to recover $\bu_t$.
For the recovery of $\bu_0, \bu_1, \bu_2$ we need to use $\bq_6, \bq_7, \bq_8$ respectively. If $t=5$ then we recover $\bv_5$ at time $\tau = 7$ and then compute $\bp_1(\bv^5),\ldots, \bp_1(\bv^7)$ and in turn recover all the missing $\bu$ symbols. If $t \in [6,8]$, then we will not be able to recover the associated $\bu_{t-6}$ but the remaining two symbols can be recovered. Thus we will have one non-recovered symbol in this case. If $t > 8$ then clearly all the three $\bu$ symbols are recovered and complete recovery is achieved.
\end{itemize}

For the case when the isolated erasure occurs at time $0$ and the burst follows it spanning the interval $[t,t+2]$, the decoder proceeds as follows.
\begin{itemize}
\item $t = 1:$ We use code $\cC_{12}$  in the interval $[0,5]$ to  recover $\bv_0,\dots,\bv_3$ at time $\tau=5$.  The decoder computes $\bp_1(\bv^t)$ in the interval $[6,9]$ and subtracts them to recover $\bu_0,\dots,\bu_3$. All the erased symbols are recovered for this erasure pattern.
\item $t \ge 2:$ We recover $\bv_0$ at time $1$ using the code $\cC_{12}$ before the start of the erasure burst.  In particular taking $\tau=1$ note that
\begin{align}
(1-R_{12})(\tau+1) = \left(1-\frac{1}{3}\right)2  > 1,
\end{align}which by Lemma~\ref{lem:mds} suffices to recover $\bv_0$ at time $\tau=1$

We next show how $\bv_t,\bv_{t+1},\bv_{t+2}$ can be recovered using the code $\cC_{12}$ in the window $[t,t+5]$. 
If $t \in \{2,3\}$ then in addition to $\bv_t, \bv_{t+1}, \bv_{t+2}$ we should also account for the erasure at time $6$
due to the repetition of $\bu_0$. Thus there are a total of $4$ erasures in the window $[t,t+5]$. An extension of Lemma~\ref{lem:mds} (c.f. Appendix~\ref{app:PRC}, Lemma~\ref{lem:twobursts}) is required to show the recovery for this pattern.
For $t \in [4,6]$ the erasure burst overlaps with the repetition of $\bu_0$ and thus there are only three erasures and the recovery follows by the application of Lemma~\ref{lem:mds}. For $t > 6$ the symbol $\bu_0$ is recovered at time $6$
and then the erasure burst  in $[t, t+2]$ can be recovered separately.

\end{itemize}
A complete justification of the above decoding steps is provided in the proof of Theorem~\ref{thm:PRC}. 

\section{Simulations Results}
\label{sec:GE}
We consider a two-state Gilbert-Elliott channel model. 
In the ``good state" each channel packet is lost with a probability of $\eps$ whereas in the ``bad state" each channel packet is lost with a probability of $1$. The average loss rate of the Gilbert-Elliott channel is given by
\begin{align}
\Pr(\cE) = \frac{\beta}{\al+\beta}\eps + \frac{\al}{\al+\beta} \label{eq:loss-uncoded}.
\end{align}
where $\al$ and $\beta$ denote the transition probability from the good state to the bad state and vice versa.

Fig.~\ref{fig:GE_T12_R5_Performance} and Fig.~\ref{fig:GE_T50_R6_Performance} show the simulation performance over a Gilbert-Elliott Channel. The parameters chosen in the two plots are as shown in Table~\ref{tab:CodeProperties_Gilbert-0}.

{\begin{table}[!htb]
	\centering
	\tabcolsep=0.09cm
		\begin{tabular}{l||ccc||ccc}
							& \multicolumn{3}{c||}{Fig.~\ref{fig:GE_T12_R5_Performance}} 	& \multicolumn{3}{c}{Fig.~\ref{fig:GE_T50_R6_Performance}} \\\hline
		$(\alpha,\beta)$ & \multicolumn{3}{c||}{$(5 \times 10^{-4},0.5)$} 			& \multicolumn{3}{c}{$(5 \times 10^{-5},0.2)$} \\
		Channel Length & \multicolumn{3}{c||}{$10^7$} 													& \multicolumn{3}{c}{$10^8$} \\
		Rate $R$ & \multicolumn{3}{c||}{$12/23 \approx 0.52$} 									& \multicolumn{3}{c}{$50/88 \approx 0.6$} \\
		Delay $T$ & \multicolumn{3}{c||}{12} 																		& \multicolumn{3}{c}{50} \\\hline
		$\eps$ 							& $10^{-3}$ & $5 \times 10^{-3}$ 	& $10^{-2}$ & $10^{-3}$ & $5 \times 10^{-3}$ 	& $10^{-2}$ \\\hline
		Burst Only 					& $0.9642$ 	& $0.8796$ 						& $0.7869$ 	& $0.9005$	& $0.5988$						& $0.3563$	\\
		Burst + Isolated		& $0.0268$ 	& $0.1065$ 						& $0.1851$ 	& $0.0923$	& $0.3065$						& $0.3698$	\\
		Burst + 							& $0.0032$ 	& $0.0081$ 						& $0.0222$ 	& $0.0062$	& $0.0937$						& $0.2729$	\\
		Multiple Isolated & & & & & &\\
		Burst Gaps $<T$						& $0.0058$ 	& $0.0058$ 						& $0.0058$ 	& $0.0010$	& $0.0010$						& $0.0010$	\\
		\end{tabular}
		\caption{Parameters of Gilbert-Elliott Channel used in Simulations. We also present the empirical fraction
		of different erasure patterns observed across a sample path.}
		\label{tab:CodeProperties_Gilbert-0}
		\vspace{-2em}
\end{table}

The channel parameters for the $T=12$ case are the same as those used in~\cite[Section 4-B, Fig.~5]{MartinianS04}.
For each of the channels we generate a sample realization and compute the residual loss rate for each of the codes in sections~\ref{sec:chan-I} and~\ref{sec:chan-2}. We also categorize the empirical fraction of different erasure patterns associated with each erasure burst for the two channels for $\eps \in \{10^{-3}, 5 \times 10^{-3}, 10^{-2}\}$ in Table~\ref{tab:CodeProperties_Gilbert-0}. The first row ``Burst-Only" denotes those erasure bursts where there are no isolated erasures in a window of length $T$ before and after the burst. The second row counts those patterns where only one isolated loss occurs in such a window. The third row allows for multiple losses in this window. The fourth row counts those patterns where the inter-burst gap is less than $T$. Note that the contribution of burst plus isolated losses is significant particularly for the second channel when the delay $T=50$.

\begin{figure*}
    \centering
    \subfigure[Simulation results.  All codes are evaluated using a decoding delay of $T=12$ symbols and a rate of $R = 12/23 \approx 0.52$.]
    {
        \includegraphics[width=0.48\linewidth, height=5cm]{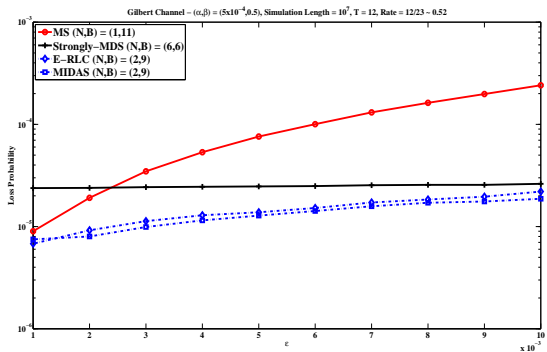}
        \label{fig:GE_T12_R5_Performance}
    }\hspace{0.1em}
    \subfigure[Burst Histogram]
    {
        \includegraphics[width=0.48\linewidth, height=5cm]{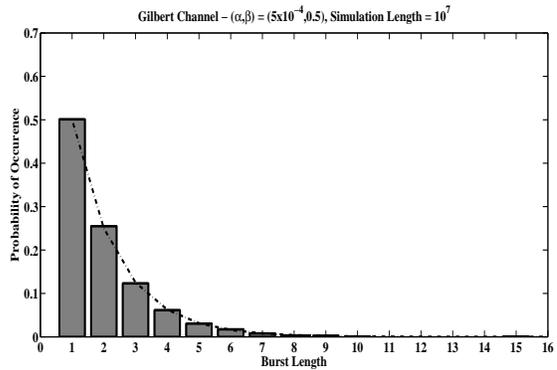}
        \label{fig:GE_T12_R5_Burst}
    }
    \caption{Simulation Experiments for Gilbert-Elliott Channel Model with  $(\al,\beta) = (5 \times 10^{-4},0.5)$.}
    \label{fig:Gilbert_T12}
\end{figure*}

\begin{figure*}
    \centering
    \subfigure[Simulation results. All codes are evaluated using a decoding delay of $T=50$ symbols and a rate of $R = 50/88 \approx 0.6$.]
    {
        \includegraphics[width=0.48\linewidth, height=5cm]{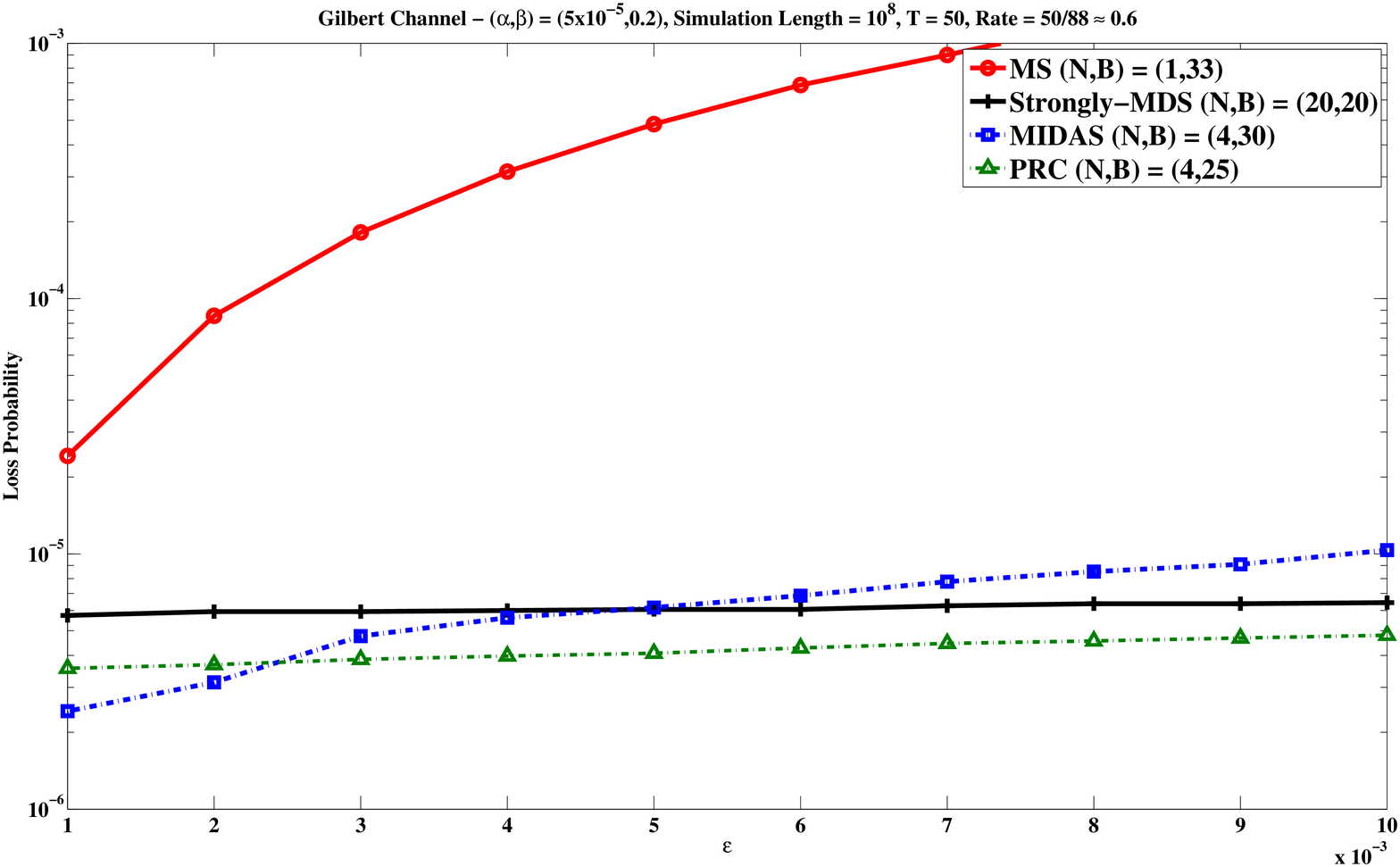}
        \label{fig:GE_T50_R6_Performance}
    }\hspace{0.1em}
    \subfigure[Burst Histogram]
    {
        \includegraphics[width=0.48\linewidth, height=5cm]{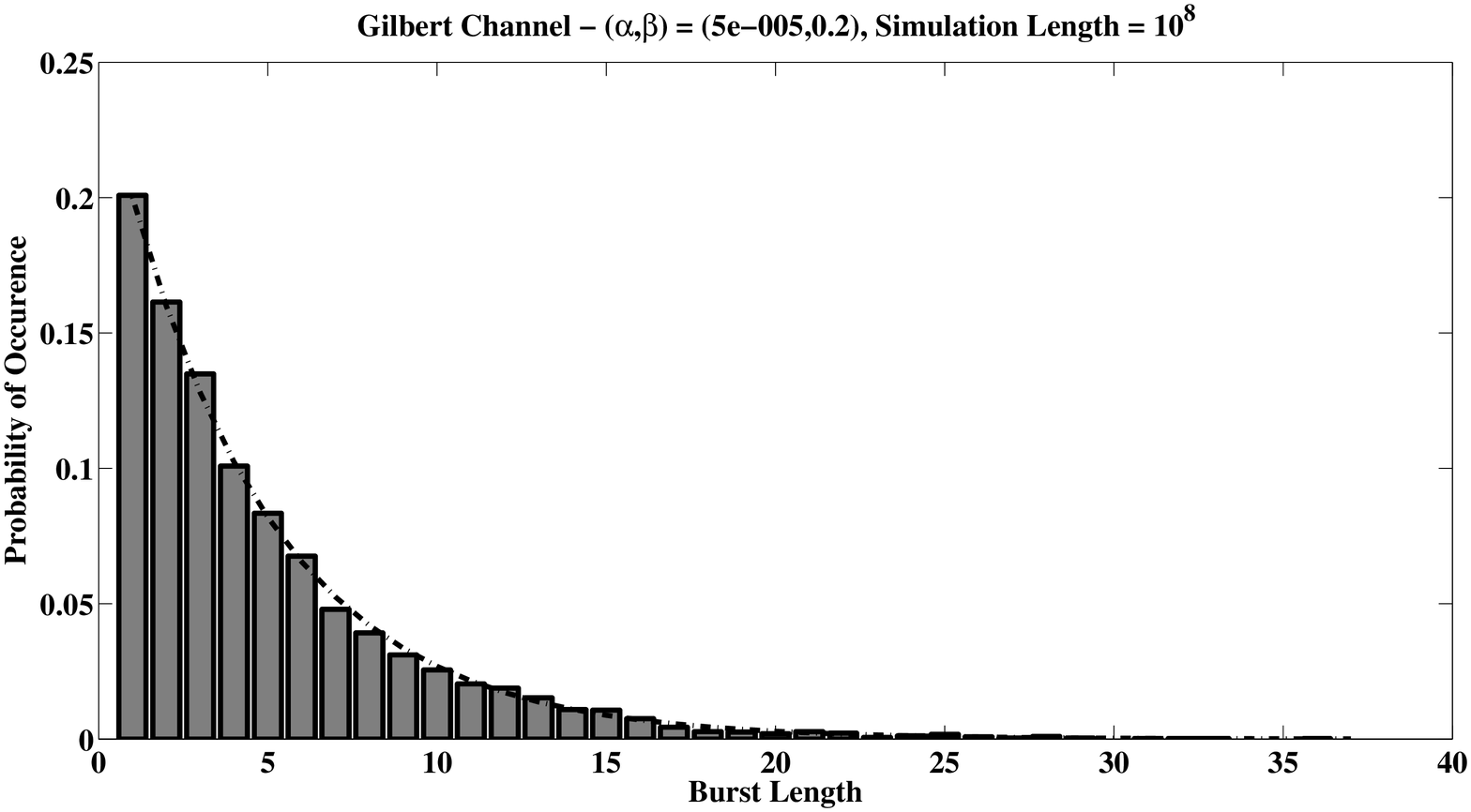}
        \label{fig:GE_T50_R6_Burst}
    }
    \caption{Simulation Experiments for Gilbert-Elliott Channel Model with  $(\al,\beta) = (5 \times 10^{-5},0.2)$.}
    \label{fig:Gilbert_T50}
\end{figure*}

In Fig.~\ref{fig:GE_T12_R5_Performance} and Fig.~\ref{fig:GE_T50_R6_Performance}, we plot the residual loss rate on the y-axis and $\eps$, which corresponds to the loss probability in the good state on the x-axis whereas Fig.~\ref{fig:GE_T12_R5_Burst}  and Fig.~\ref{fig:GE_T50_R6_Burst} illustrate the burst histograms associated with the GE channels. 

All codes in Fig.~\ref{fig:GE_T12_R5_Performance} are selected to have a rate of $R=12/23 \approx 0.52$ and the delay is $T=12$. The black horizontal line is the loss rate of the Strongly-MDS code. It achieves $B=N=6$. Thus its performance is limited by its burst-correction capability. The red-curve which deteriorates rapidly as we increase $\eps$ is the Maximally Short codes (MS). It achieves $B=11$ and $N=1$. Thus in general it cannot recover from even two losses occurring in a window of length $T+1$.  The remaining two plots correspond to the E-RLC codes and MiDAS codes, both of which achieve $B=9$ and $N=2$. The loss probability also deteriorates with $\eps$ for both codes, although at a much lower rate. Thus a slight decrease in $B$, while improving $N$ from $1$ to $2$ exhibits noticeable gains over both MS and Strongly-MDS codes.

In Fig.~\ref{fig:GE_T50_R6_Performance} the rate of all codes except PRC is set to $R=50/88\approx 0.57$ while the rate of the PRC codes is $0.6$. The delay is set to $T=50$. The Strongly-MDS codes achieve ${B=N=20}$ whereas the MS codes achieve $N=1$ and $B=33$. Both codes suffer from the same phenomenon discussed in the previous case. We also consider the MiDAS code with $N=4$ and $B=30$. We observe that its performance deteriorates as $\eps$ is increased, mainly because it fails on burst plus isolated loss patterns, which are dominant in this regime (see Table~\ref{tab:CodeProperties_Gilbert-0}). The PRC code which achieves $N=4$ and $B=25$ and can also handle burst plus one isolated loss exhibits the best performance in this plot. 

Interestingly Table~\ref{tab:CodeProperties_Gilbert-0} also indicates that the probability of having burst plus multiple isolated losses associated with the burst is also significant for the second GE channel with $T=50$. The PRC code will not handle more than one isolated loss associated with the erasure burst. Nevertheless we did not notice a performance degradation due to this in the simulations. To explain this, note that for the GE channel of interest the average burst length is only $\frac{1}{\beta} = 5$
and as shown in Fig.~\ref{fig:GE_T50_R6_Burst} the burst length distribution is an exponential function. Our proposed PRC code uses a shift of $\Delta = 46$ (c.f. section~\eqref{sec:PRC-construction}). It turns out that the isolated erasures which result in decoding failures must be concentrated in the interval $[\Delta-1, \Delta+B-1]$ following the burst rather than the entire interval of length $T$. As $B$ is generally small, it is unlikely that more than one isolated loss is seen in this interval, and hence the effect of multiple isolated losses does not appear to be significant.  Similarly for the first GE channel, since the average burst length is only $2,$  the use of PRC codes to recover from burst plus isolated losses was not found to be effective.

\section{Conclusion}
\label{sec:concl}
We consider new low delay streaming code constructions based on a layered design principle. Our proposed codes are designed for a class of packet-erasure channels that are restricted to introduce a certain set of erasure patterns. These patterns correspond to the dominant set of error events associated with the Gilbert-Elliott channel.  The first class of channels introduce either an erasure burst or a certain number of isolated erasures in any sliding window of interest. We show that there exists a fundamental tradeoff between the burst-error correction and isolated-error correction capability of any streaming code for such channels and propose a class of codes --- MiDAS Codes --- that achieve a near optimal tradeoff. In these constructions, we first construct an optimal code for the burst-erasure channel and then concatenate an additional layer of parity checks for the isolated erasures. We also consider a more complex channel model that involves both a burst erasure and an isolated erasure,  and propose streaming codes that recover all but one source symbol for each such erasure pattern . 
Numerical simulations over the Gilbert Elliott channel indicate that the proposed codes indeed provide significant performance gains over baseline codes. 

\appendices

\section{Proof of Lemma~\ref{lem:mds}}
\label{app:MDS}

We begin by introducing the $j-$th column distance $d_j^c$ of an $(n,k,m)$ convolutional code. Consider the window of the first $j+1$ symbols and let the truncated codeword associated with the input sequence $(\bs[0], \ldots, \bs[j])$ be $(\bx[0],\dots,\bx[j])$. Then the $j$-th column distance is defined as
\begin{align}
\label{eq:djc}
d_j^c = \min_{\substack{\bs \equiv (\bs[0], \ldots, \bs[j])\\ \bs[0] \neq 0}} \mathrm{wt}^c(\bx[0],\dots,\bx[j]),
\end{align}
where recall that each channel symbol $\bx[i]$ has $n$ sub-symbols, i.e., $\bx[i]=(x_0[i],\dots,x_{n-1}[i])$ and $\mathrm{wt}^c(\bv)$ counts the number of non-zero \textbf{sub-symbols} in the codeword $\bv$.

It is well-known that for any $(n,k,m)$ convolutional code  $d_j^c \leq (n-k)(j+1)+1$ for all $j \ge 0$. A special class of  convolutional codes -- systematic Strongly-MDS codes --  satisfy this bound with equality for $j=\{0,\dots,m\}$~\cite[Corollary 2.5]{strongly-mds}. In our setup, since the entire channel packet is either transmitted perfectly or erased, it is more convenient to express the column-distance by counting the number of non-zero symbols. We define the column-distance $d_j$ as the minimum number of non-zero \textbf{symbols} in the ${j+}1$ truncated codeword $(\bx[0],\dots,\bx[j])$, i.e.,
\begin{align}
\label{eq:dj}
d_j = \min_{\substack{\bs \equiv (\bs[0], \ldots, \bs[j])\\ \bs[0] \neq 0}} \mathrm{wt}(\bx[0],\dots,\bx[j]),
\end{align}
where $\mathrm{wt}(\bv)$ counts the number of non-zero \textbf{symbols} in the codeword $\bv$.

Clearly for any $(n,k,m)$ code, we must have that\footnote{We use $\lceil{v}\rceil$ and $\lfloor{v}\rfloor$ to denote the ceil and floor of a real number $v$.}  $d_j \ge \left\lceil\frac{d_j^c}{n}\right\rceil$.   Thus for a Strongly-MDS code we have that
\begin{align}
d_j &\ge \left \lceil {\frac{(n-k)(j+1)+1}{n}} \right \rceil \nonumber \\
		&= \left \lceil {(1-R)(j+1) + \frac{1}{n}} \right \rceil \nonumber \\
		&\geq \left \lfloor {(1-R)(j+1)} \right \rfloor + 1 \label{eq:dj-bnd}
\end{align}
where the last relation is justified as follows. 
If $v = (1-R)(j+1)$ is an integer, $\left\lceil {v+\frac{1}{n}}\right\rceil = v + 1 = \lfloor{v}\rfloor +1$, and thus the last step follows with equality. 
Otherwise, we have  $\left\lceil {v+\frac{1}{n}}\right\rceil \ge \left\lceil {v}\right\rceil = \lfloor{v}\rfloor +1$ and the last step still holds.

To establish P1,  consider the interval $[0,j]$ and consider two input sequences $(\bs[0],\ldots, \bs[j])$ and $(\bs'[0],\ldots, \bs'[j])$ with $\bs[0] \neq \bs'[0]$. Let the corresponding output be $(\bx[0],\ldots, \bx[j])$ and $(\bx'[0],\ldots, \bx'[j])$. Note that the output sequences differ in at-least $d_j$ symbols since otherwise the output sequence $(\bx[0] - \bx'[0], \ldots, \bx[j]-\bx'[j])$ which corresponds to $(\bs[0]-\bs'[0], \ldots, \bs[j]-\bs'[j])$ has a Hamming weight less than $d_j$ while the input $\bs[0] -\bs'[0] \neq 0$, which is a contradiction. 
Thus if $(\bs[0],\ldots, \bs[j])$ is transmitted, for any sequence of $d_j-1$ erased symbols, there will be at-least one symbol where $(\bx'[0], \ldots, \bx'[j])$ differs from the received sequence. Thus from~\eqref{eq:dj-bnd}, $\bs[0]$ is recovered uniquely at time $j$ provided that:
\begin{align}
N \le \left \lfloor {(1-R)(j+1)} \right \rfloor 
\end{align}
from which P1 in Lemma~\ref{lem:mds} immediately follows.


\begin{figure}[htbp]
	\centering
	\resizebox{0.5\columnwidth}{!}{
	\begin{tikzpicture}[node distance=0mm]
		\node                       (x1start) {Link:};
		\node[esym, right = of x1start] (x100) {};
		\node[esym, right = of x100]     (x101) {};
		\node[esym, right = of x101]     (x102) {};
		\node[esym, right = of x102]     (x103) {};
		\node[esym, right = of x103]     (x104) {};
		\node[usym, right = of x104]     (x105) {};
		\node[usym, right = of x105]     (x106) {};
		\node[usym, right = of x106]     (x107) {};
		\node[usym, right = of x107]     (x108) {};
		\node[usym, right = of x108]     (x109) {};
		\node[usym, right = of x109]     (x110) {};
		\node[usym, right = of x110]     (x111) {};
		\node[esym, right = of x111]     (x112) {};
		\node[esym, right = of x112]     (x113) {};
		\node[esym, right = of x113]     (x114) {};
		\node[esym, right = of x114]     (x115) {};
		\node[esym, right = of x115]     (x116) {};
		\node[usym, right = of x116]     (x117) {};
		\node[usym, right = of x117]     (x118) {};
		\node[usym, right = of x118]     (x119) {};
		\node[usym, right = of x119]     (x120) {};
		\node[usym, right = of x120]     (x121) {};
		\node[usym, right = of x121]     (x122) {};
		\node[usym, right = of x122]     (x123) {};
		\node[esym, right = of x123]     (x124) {};
		\node      [right = of x124]   (x1end) {$\cdots$};

		\braceup{x100}{x111}{1mm}{$W_0$};
		\braceup{x104}{x115}{3mm}{$W_{B-1}$};
		
		\dimdn{x100}{x104}{-2mm}{$B$};
		\dimdn{x100}{x111}{-7mm}{Period = $j+1$};
	\end{tikzpicture}}
	\caption{The periodic erasure channel used in proving P2 in Lemma~\ref{lem:mds}, and indicating the first and last windows of interest, $W_0$ and $W_{B-1}$, respectively. Grey and white squares resemble erased and unerased symbols respectively.}
	\label{fig:mdp-oneburst}
\end{figure}
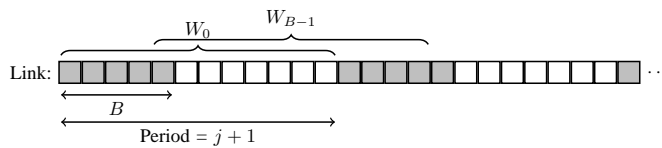

In order to establish property P2, we consider a periodic erasure channel of period length ${j+1}$ as shown in Fig.~\ref{fig:mdp-oneburst}. 
In each period, a burst of length $B$ symbols is introduced followed by $j-B+1$ un-erased symbols. 
First consider the interval $W_0 = [0,j]$. Using the previous part and since $B \le (1-R)(j+1)$ we are guaranteed that $\bs[0]$ is recovered at time $t=j$.
At this point we have recovered $\bs[0]$ and treat it as a non-erased symbol. We next consider the interval $W_1 = [1,j+1]$. There is a burst of length $B-1$
spanning the period $[1,B]$ and an additional erasure at time $\bx[j+1]$. By property P1 we can recover $\bs[1]$ by time $t=j+1$. However since $\bx[j+1]$
is erased, this is equivalent to recovering $\bs[1]$ at time $t=j$.  In a similar fashion by considering the intervals
$W_i \triangleq [i,i+j]$ for $i = \{2,\dots,B-1\}$ we recover $\bs[2],\bs[3],\dots,\bs[B-1]$ at time $j$. At this point all the erased symbols in the first burst have been recovered and the claim follows.


\section{Proof of Decoding in Theorem~\ref{thm:PRC}}
\label{app:PRC}

We first begin with the following extension of Lemma~\ref{lem:mds} which considers two bursts in the interval $[0,j]$.

\begin{lemma}
\label{lem:twobursts}
In the interval $[0,j]$ suppose that there are two erasure bursts spanning the intervals $[0,B_1-1]$ and $[r,r+B_2-1]$. 
 A $(n,k,m)$ Strongly-MDS convolutional code with rate $R = \frac{k}{n},$ can recover all the erased ${B_1+B_2}$ source symbols $\bs[0],\dots,\bs[B_1-1],\bs[r],\dots,\bs[r+B_2-1]$ by time $t=j$ provided that
\begin{align}
r \le \frac{B_1}{1-R} \quad \text{and} \quad B_1+B_2 \leq (1-R)(j+1) \label{eq:twobursts-cond}
\end{align}
 for any $j=0,1,\ldots, m$.
\end{lemma}
\begin{figure}[htbp]
	\centering
	\resizebox{0.5\columnwidth}{!}{
	\begin{tikzpicture}[node distance=0mm]
		\node                       (x1start) {Link:};
		\node[esym, right = of x1start] (x100) {};
		\node[esym, right = of x100]     (x101) {};
		\node[esym, right = of x101]     (x102) {};
		\node[esym, right = of x102]     (x103) {};
		\node[usym, right = of x103]     (x104) {};
		\node[usym, right = of x104]     (x105) {};
		\node[esym, right = of x105]     (x106) {};
		\node[esym, right = of x106]     (x107) {};
		\node[esym, right = of x107]     (x108) {};
		\node[usym, right = of x108]     (x109) {};
		\node[usym, right = of x109]     (x110) {};
		\node[usym, right = of x110]     (x111) {};
		\node[esym, right = of x111]     (x112) {};
		\node[esym, right = of x112]     (x113) {};
		\node[esym, right = of x113]     (x114) {};
		\node[esym, right = of x114]     (x115) {};
		\node[usym, right = of x115]     (x116) {};
		\node[usym, right = of x116]     (x117) {};
		\node[esym, right = of x117]     (x118) {};
		\node[esym, right = of x118]     (x119) {};
		\node[esym, right = of x119]     (x120) {};
		\node[usym, right = of x120]     (x121) {};
		\node[usym, right = of x121]     (x122) {};
		\node[usym, right = of x122]     (x123) {};
		\node[esym, right = of x123]     (x124) {};
		\node      [right = of x124]   (x1end) {$\cdots$};

		\braceup{x100}{x111}{1mm}{$W_0$};
		\braceup{x103}{x114}{3mm}{$W_{B_1-1}$};
		
		\dimdn{x100}{x103}{-2mm}{$B_1$};
		\dimdn{x106}{x108}{-2mm}{$B_2$};
		\dimdn{x100}{x111}{-7mm}{Period = $j+1$};
	\end{tikzpicture}}
	\caption{The periodic erasure channel used in proving Lemma~\ref{lem:twobursts}, and indicating the first and last windows of interest, $W_0$ and $W_{B-1}$, respectively. Grey and white squares resemble erased and unerased symbols respectively.}
	\label{fig:mdp-twobursts}
\end{figure}
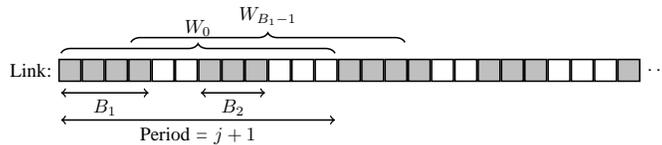
\begin{proof}
Consider a periodic erasure channel of period length $j+1$ as shown in Fig.~\ref{fig:mdp-twobursts}. 
In each period, the channel introduces an erasure burst of length $B_1$ symbols followed by another burst of length $B_2$ symbols starting $r$ symbols from the start of the first burst and the rest of the period is not erased. 
In the first period starting at time $t=0$, the two bursts of length $B_1$ and $B_2$ span the intervals $[0,B_1-1]$ and $[r,r+B_2-1]$ respectively, where $r < \frac{B_1}{1-R}$. 

We consider the intervals $W_i = [i,i+j-1]$ for $i = \{0,1,\dots,B_1-1\}$. In each such interval, the total number of erasures equals $B_1+B_2 \leq (1-R)(j+1)$. Thus as in the proof of property P2 in Lemma~\ref{lem:mds}, we can  recover $\bs[0],\bs[1],\dots,\bs[B_1-1]$ by time $j$ since $\bx[j+1],\bx[j+2],\dots,\bx[j+B_1-1]$ are erased. 

For the second burst we consider the interval $[r,j]$ of length $j-r+1$. At this point all the source symbols in the first burst have been recovered and only a total of $B_2$ erasures remain. The property P2 in Lemma~\ref{lem:mds} can be used to for recovering $\bs[r],\dots,\bs[r+B_2-1]$ by time $j$ since,
\begin{align}
B_2 &\leq (1-R)(j+1) - B_1 \nonumber \\
		&\le (1-R)(j+1) - (1-R)r \nonumber \\
		&= (1-R)(j-r+1) < d_{j-r}.
\end{align}
where the second step uses~\eqref{eq:twobursts-cond} and the last step follows using~\eqref{eq:dj-bnd}. At this point all the erased symbols in the first period have been recovered by time $j$ and the claim follows.

\end{proof}

\usetikzlibrary{patterns}
\begin{figure}[t]
\subfigure[Burst followed by Isolated Erasure: Burst and Isolated Erasures Recovered Simultaneously]
{\resizebox{\columnwidth}{!}{
\begin{tikzpicture}[node distance=1mm,]
  \tikzstyle{triple2} = [rectangle split, anchor=text,rectangle split parts=4]
  \tikzstyle{double2} = [rectangle split, anchor=text,rectangle split parts=2]
	\tikzstyle{single2} = [rectangle split, anchor=text,rectangle split parts=1]
  \tikzstyle{triple} = [draw, rectangle split,rectangle split parts=4]
	\tikzstyle{double} = [draw, rectangle split,rectangle split parts=2]
	\tikzstyle{single} = [draw, rectangle split,rectangle split parts=1]
	\tikzset{block/.style={rectangle,draw}}
	\tikzset{block2/.style={rectangle}}
	
	\input{partial-recoveryA1}
	\label{fig:Partial_CaseA1}

\end{tikzpicture}}
}
\vspace{1em}

\subfigure[Burst followed by Isolated Erasure: Burst and Isolated Erasures Recovered Separately. $\tau = \left\lceil {B\frac{ \Delta}{B+1}}\right\rceil -1$]
{\resizebox{\columnwidth}{!}{
\begin{tikzpicture}[node distance=1mm,]
  \tikzstyle{triple2} = [rectangle split, anchor=text,rectangle split parts=4]
  \tikzstyle{double2} = [rectangle split, anchor=text,rectangle split parts=2]
	\tikzstyle{single2} = [rectangle split, anchor=text,rectangle split parts=1]
  \tikzstyle{triple} = [draw, rectangle split,rectangle split parts=4]
	\tikzstyle{double} = [draw, rectangle split,rectangle split parts=2]
	\tikzstyle{single} = [draw, rectangle split,rectangle split parts=1]
	\tikzset{block/.style={rectangle,draw}}
	\tikzset{block2/.style={rectangle}}
	
	\input{partial-recoveryA2}
	\label{fig:Partial_CaseA3}

\end{tikzpicture}}
}

\vspace{1em}

\subfigure[Isolated Erasure followed by Burst: Isolated and Burst Erasures Recovered Simultaneously]
{\resizebox{\columnwidth}{!}{
\begin{tikzpicture}[node distance=1mm,]
  \tikzstyle{triple2} = [rectangle split, anchor=text,rectangle split parts=4]
  \tikzstyle{double2} = [rectangle split, anchor=text,rectangle split parts=2]
	\tikzstyle{single2} = [rectangle split, anchor=text,rectangle split parts=1]
  \tikzstyle{triple} = [draw, rectangle split,rectangle split parts=4]
	\tikzstyle{double} = [draw, rectangle split,rectangle split parts=2]
	\tikzstyle{single} = [draw, rectangle split,rectangle split parts=1]
	\tikzset{block/.style={rectangle,draw}}
	\tikzset{block2/.style={rectangle}}
	
	\input{partial-recoveryB1}
	\label{fig:Partial_CaseB1}

\end{tikzpicture}}
}

\vspace{1em}

\subfigure[Isolated Erasure followed by Burst: Isolated and Burst Erasures Recovered Separately. $\tau = \left\lceil {\frac{ \Delta}{B+1}}\right\rceil -1$]
{\resizebox{\columnwidth}{!}{
\begin{tikzpicture}[node distance=1mm,]
  \tikzstyle{triple2} = [rectangle split, anchor=text,rectangle split parts=4]
  \tikzstyle{double2} = [rectangle split, anchor=text,rectangle split parts=2]
	\tikzstyle{single2} = [rectangle split, anchor=text,rectangle split parts=1]
  \tikzstyle{triple} = [draw, rectangle split,rectangle split parts=4]
	\tikzstyle{double} = [draw, rectangle split,rectangle split parts=2]
	\tikzstyle{single} = [draw, rectangle split,rectangle split parts=1]
	\tikzset{block/.style={rectangle,draw}}
	\tikzset{block2/.style={rectangle}}
	
	\input{partial-recoveryB2}
	\label{fig:Partial_CaseB3}

\end{tikzpicture}}
}
\caption{Various  erasure patterns of the Channel $\rm{II}$ codes in the decoding analysis. The erasures are  shaded grey boxes, the parity check symbols used to recover the $\bv[\cdot]$ symbols are marked using diagonal stripes whereas the parity check symbols involved in the recovery of the $\bu[\cdot]$ symbols are marked using horizontal stripes.}
\label{fig:Partial_Cases}
\end{figure}

We now proceed to prove Theorem~\ref{thm:PRC}. For any given erasure pattern, in which the channel introduces a burst of length $B$ and one isolated erasure in a sliding window of length $2T+B$, the decoder can recover all symbols with a delay of $T$ but one. We divide these patterns into two main categories. In the first case the erasure burst is followed by an isolated erasure whereas in the second case an isolated erasure preceeds the erasure burst.  


\subsection{Erasure Burst followed by an Isolated Erasure}
Without loss of generality assume that the channel introduces an erasure burst in the interval $[0,B-1]$ and that the isolated erasure occurs at time $t \geq B$. Since the associated isolated erasure follows the erasure burst from Def.~\ref{def:associated} it must occur in the interval $[B, T+B-1]$. This implies that the interval $[-T,-1]$ is free of any erasure so that there is only one burst and isolated erasure in the interval $[-T, T+B-1]$, which is of length $2T+B$. Since the memory of the code equals $T$, any erased
symbols before $t < -T$ will not affect the decoder. Thus we assume that there are no erasures before $t=0$.

We further consider two cases as stated below.

\subsubsection*{{\bf Burst and Isolated Erasures Recovered Simultaneously}}
In this case,  the burst and the isolated erasures are close enough such that all the $\bv[\cdot]$ symbols are recovered simultaneously. This case is illustrated in Fig.~\ref{fig:Partial_CaseA1}. The isolated erasure happens at time $t$ where $B \le t < \frac{B}{B+1}\Delta$. The recovery of the erased symbols proceeds as follows:
\begin{enumerate}
\item Recover $\{\bv[0],\ldots, \bv[B-1], \bv[t]\}$ at time $\tau=\Delta-1$ using the $(v+u+s,v)$ Strongly-MDS code $\cC_{12}$ in the interval $[0,\Delta-1]$.
\item Recover $\{\bu[0],\ldots, \bu[B-1], \bu[t]\}$ at time $\tau= \Delta, \ldots, \Delta+B-1$ and $\tau = t+\Delta$ respectively from the associated parity checks $\bq[\cdot]$.
\end{enumerate}

To justify the recovery of $\bv[0],\ldots, \bv[B-1], \bv[t]$ in the first step we consider the available parity checks of $\cC_{12}$ in the interval $[0,\Delta-1]$.
We first note that the interfering symbols $\bu[\cdot]$ in this interval are  not erased and can be cancelled out from $\bq[\cdot]$ to recover the parity checks $\bp_1[\cdot]$.
We apply Lemma~\ref{lem:twobursts} with $B_1=B,$ $B_2=1$, $r=t$ and $R = R_{12}$ and $j=\Delta-1$. Note that
$t < \frac{B}{B+1}\Delta$ also satisfies $t < \frac{B_1}{1-R_{12}}$ since $R_{12} = \frac{\Delta-B-1}{\Delta}$ from~\eqref{eq:r12-PRC}. Thus the first condition in~\eqref{eq:twobursts-cond} in Lemma~\ref{lem:twobursts} is satisfied. Furthermore note that $(1-R_{12})(j+1)=B+1$ and thus the second condition in~\eqref{eq:twobursts-cond} in Lemma~\ref{lem:twobursts} is also satisfied. Thus Lemma~\ref{lem:twobursts} applies and the recovery of $\bv[0],\ldots, \bv[B-1],\bv[t]$ at time $\Delta-1$ follows. 

To justify the recovery of $\{\bu[0],\ldots, \bu[B-1], \bu[t]\}$, recall that $\bq[i] = \bu[i-\Delta] + \bp_1[i]$. Since all the $\bv[\cdot]$ symbols have been recovered in step (1), the associated parity checks $\bp_1[i]$ can be computed and cancelled by the decoder to recover the $\bu[\cdot]$ symbols as claimed.

As a final remark we note that all the erased symbols are recovered for the above erasure pattern.

\subsubsection*{{\bf Burst and Isolated Erasures Recovered Separately}}
In this case, there is a sufficiently large gap between the burst and the isolated erasures so that the $\bv[\cdot]$ symbols of the erasure burst are recovered before the isolated erasure takes place. This case is illustrated in Fig.~\ref{fig:Partial_CaseA3}. The isolated erasure happens at time $t\geq B\frac{ \Delta}{B+1}$. 
The recovery of the erased symbols proceeds as follows:
\begin{enumerate}
\item Recover $\{\bv[0],\ldots, \bv[B-1]\}$ by time $\tau = \left\lceil {B\frac{ \Delta}{B+1}}\right\rceil -1$ using the $(v+u+s,v)$ Strongly-MDS code $\cC_{12}$ in the interval $[0,\tau]$.
\item Recover $\{\bu[0],\ldots, \bu[t-\Delta-1]\}$ from $\bq[\Delta],\ldots, \bq[t-1]$ respectively by cancelling the interfering $\bp_1[\cdot]$ symbols.
\item Recover $\bv[t]$ by time $\tau = t+T-\Delta+1$ using the $(v+s, v)$ Strongly-MDS code $\cC_2$ in the interval $[t, t+T-\Delta+1]$.
\item Recover $\{\bu[t-\Delta+1], \ldots, \bu[B-1], \bu[t]\}$ from $\bq[t+1],\ldots, \bq[B+\Delta-1], \bq[t+\Delta]$ by cancelling the interfering $\bp_1[\cdot]$ symbols.
\end{enumerate}

To justify the recovery of  $\{\bv[0],\ldots, \bv[B-1]\}$ in the first step above, we consider the available parity checks of $\cC_{12}$ in the interval $[0,\tau]$.
Note that the interfering $\bu[\cdot]$ symbols in this interval are not erased and can be cancelled out from $\bq[\cdot]$ to recover the underlying parity checks $\bp_1[\cdot]$. Furthermore,
\begin{align}
(1-R_{12})(\tau+1) \ge (1-R_{12})\left(\frac{B}{B+1}\Delta\right) = B
\end{align}
where we substituted~\eqref{eq:r12-PRC} for $R_{12}$. Thus using property P2 in Lemma~\ref{lem:mds} we recover the $\bv[0],\dots,\bv[B-1]$ by time $\tau$
 as stated. To justify step (2) above note that $\bq[i] = \bu[i-\Delta] + \bp_1[i]$ and the interfering $\bp_1[i]$ are only functions of $\bv[\cdot]$ symbols that have either been recovered in step (1) or are not erased. To justify (3) we consider the interval $[t, t+T-\Delta+1]$ and consider the parity checks of $\cC_2$ in this interval. Note that using
\eqref{eq:r2-PRC} we have:\begin{align}
(1-R_2)(T-\Delta+2) \ge 1
\end{align}holds and hence using property P2 in Lemma~\ref{lem:mds} we recover $\bv[t]$ by time $t+T-\Delta+1$. To justify step (4), note that once $\bv[t]$ is recovered in step (3),  the parity checks $\bp_1[t+1],\ldots, \bp_1[B+\Delta], \bp_1[t+\Delta]$ can be computed and cancelled from the associated $\bq[\cdot]$ symbols, and the claim follows.

As a final remark, we note that when if $t \in [\Delta,\Delta+B-1]$ the symbol $\bu[t-\Delta]$ which is erased in the first burst, cannot be recovered  as its repeated copy at time $t$ is also erased. This is the only symbol that cannot be recovered in the above erasure pattern.

\subsection{Isolated Erasure followed by an Erasure Burst}

We assume without loss of generality that the isolated erasure happens at time zero and that the burst erasure happens at time $t >0$. Since the isolated erasure precedes the erasure burst, it follows that the erasure burst must begin in the interval $t \in [1, T]$ from Def.~\ref{def:associated}. This implies that there cannot be any erasure in the interval $[-T, -1]$ since the interval $[-T, T+B-1]$ must have only one isolated erasure and one erasure burst. Since the memory of the code equals $T$, any erased symbol before time $t = -T$ will not affect the decoder. Thus in what follows we assume that there are no erasures before time $0$.

 This class of patterns is sub-divided into two cases discussed below.

\subsubsection*{{\bf Isolated and Burst Erasures Recovered Simultaneously}}
In this case the burst erasure and the isolated erasure are close enough so that all the $\bv[\cdot]$ symbols are simultaneously recovered. This case is illustrated in Fig.~\ref{fig:Partial_CaseB1}. The burst erasure begins at time $t < \frac{\Delta}{B+1}$. The recovery of the erased symbols proceeds as follows:
\begin{enumerate}
\item Recover $\{\bv[0],\bv[t],\ldots, \bv[t+B-1]\}$ using the $(v+u+s,v)$ Strongly-MDS code $\cC_{12}$ in the interval $[0,\Delta-1]$.
\item Recover $\{\bu[0], \bu[t],\ldots, \bu[t+B-1]\}$ at time $\tau = \Delta, t+\Delta, \ldots, t+B+\Delta-1$ respectively from the associated parity checks $\bq[\cdot]$. 
\end{enumerate}

To justify step (1) we note that the interfering $\bu[\cdot]$ symbols in $\bq[\cdot]$ in the interval $[0,\Delta-1]$  are not erased and can be cancelled to recover $\bp_1[\cdot]$. We apply Lemma~\ref{lem:twobursts} to code $\cC_{12}$ in the interval $[0,\Delta-1]$ using $B_1=B,$ $B_2=B$ and $r=t$. Note that by assumption on $t$ and from~\eqref{eq:r12-PRC} we have that
\begin{align}
r < \frac{\Delta}{B+1} = \frac{1}{1-R_{12}}
\end{align}
and thus the first condition in~\eqref{eq:twobursts-cond} holds.  Furthermore from~\eqref{eq:r12-PRC} we also have that $(1-R_{12})\Delta = B+1$ and thus the second condition in~\eqref{eq:twobursts-cond} also holds. Thus  Lemma~\ref{lem:twobursts} guarantees the recovery of $\{\bv[0],\bv[t],\ldots, \bv[t+B-1]\}$ by time $\tau=\Delta-1$.

To justify step (2), note that there are no further erasures in the interval $[\Delta, \Delta +t + B-1]$. Since all the erased $\bv[\cdot]$
symbols are recovered in step (1), the decoder can compute $\bp_1[\Delta],\bp_1[i+\Delta],\dots,\bp_1[i+B+\Delta-1]$ and subtract them from the corresponding $\bq[\cdot]$ symbols to recover $\bu[0],\bu[i],\dots,\bu[i+B-1]$, respectively with a delay of $\Delta \leq T$.

As a final remark we note that all the erased symbols are fully recovered in this erasure pattern.

\subsubsection*{{\bf Isolated and Burst Erasures Recovered Separately}}
In this case the gap between the isolated erasure and the burst erasure is sufficiently large so that  $\bv[0]$ is recovered before the burst erasure begins.  This case is illustrated in Fig.~\ref{fig:Partial_CaseB3}. In this case we have that $t \ge \frac{\Delta}{B+1}$.
The recovery of the erased symbols proceeds as follows:
\begin{enumerate}
\item Recover the symbol $\bv[0]$ by time $\tau = \left\lceil\frac{\Delta}{B+1}\right\rceil-1$ using the $(v+u+s, v)$ Strongly-MDS code $\cC_{12}$
in the interval $[0,\tau]$.
\item Recover the symbols $\bv[t],\ldots, \bv[t+B-1]$ by time $t+\Delta-1$ using the $(v+u+s, v)$ Strongly-MDS code $\cC_{12}$
in the interval $[t,t+\Delta-1]$.
\item Recover $\bu[t],\ldots, \bu[t+B-1]$ from $\bq[t+\Delta],\ldots, \bq[t+B+\Delta-1]$ respectively by cancelling the associated $\bp_1[\cdot]$ symbols. 
\end{enumerate}
To justify the above steps note the interfering $\bu[\cdot]$ symbols in $\bq[\cdot]$ for $t \in [0, \Delta-1]$ are not erased and can be cancelled out to recover $\bp_1[\cdot]$.   In step (1), it suffices to use P1 in Lemma~\ref{lem:mds} and show that $\bv[0]$ is recovered by time $\tau = \left\lceil\frac{\Delta}{B+1}\right\rceil-1$. Note that
\begin{align}
(1-R_{12})(\tau+1) \ge (1-R_{12})\frac{\Delta}{B+1} =1
\end{align}
where we substitute~\eqref{eq:r12-PRC} for $R_{12}$ above. Since by assumption on $t$, $\bv[0]$ is the only symbol erased in the interval $[0,\tau]$ it follows that $\bv[0]$ is recovered by this time.

To justify step (2), consider the interval $[t, t+\Delta-1]$ and recall that the erasure burst spans $[t,t+B-1]$. Furthermore even though $\bv[0]$ has been recovered in step (1) and its effect can be cancelled out, the symbol $\bu[0]$ appears in $\bq[\Delta]$ and may contribute to one additional erasure when $t \le \Delta$. In this case, we assume that a total of ${B+1}$ erasures occur in the above stated interval. We use Lemma~\ref{lem:twobursts} applied to the code $\cC_{12}$ with $B_1=B$ and $B_2=1$, in order to show the recovery of $\bv[t],\ldots, \bv[t+B-1]$. Note that the first condition in~\eqref{eq:twobursts-cond} is satisfied since 
\begin{align}
\Delta - t \le \Delta \frac{B}{B+1} = \frac{B}{1-R_{12}}
\end{align}
is satisfied and the second condition is satisfied as well since $(1-R_{12})\Delta = B+1$. By  time $t+\Delta-1$, the decoder has recovered all the erased $\bv[\cdot]$ symbols. If instead we had $t > \Delta$ then $\bu[0]$ can be recovered at time $t=\Delta$ and there remain only $B$ erasures in the interval $[t,t+\Delta-1]$, so the recovery of $\bv[t],\ldots, \bv[t+B-1]$ again follows. 

Finally to recover the $\bu[\cdot]$ symbols in the interval $[t,t+\Delta-1],$ we compute the parity check symbols $\bp_1[\cdot]$ in the interval $[t+\Delta,t+B+\Delta-1],$  subtract them from the corresponding $\bq[\cdot]$ symbols, and recover $\bu[t],\dots,\bu[t+B-1]$ respectively as stated in step (4).

As a final remark we note that the symbol $\bu[0]$ my not be recovered if its repeated copy at time $\Delta$ is erased as part of the erasure burst. Thus we may have one unrecovered symbol for the above erasure pattern.

This completes the proof of the decoder in Theorem~\ref{thm:PRC}.


\bibliographystyle{IEEEtran}
\bibliography{sm}

\end{document}